\documentclass[12pt]{article}

\usepackage{amssymb,hyperref, amsmath,amsfonts, bm, eurosym,geometry,ulem,graphicx,caption,mathtools,natbib, xcolor,setspace,sectsty,comment,footmisc,caption,pdflscape,subfigure,array, parskip, bm, csquotes, multirow, booktabs,caption}

\usepackage[flushleft]{threeparttable}

\usepackage[capposition=top]{floatrow}

\usepackage[capposition=top]{floatrow}
\usepackage[toc,page]{appendix}

\graphicspath{ {./images/} }

\DeclarePairedDelimiter\floor{\lfloor}{\rfloor}

\newtheorem{theorem}{Theorem}

\newtheorem{proposition}{Proposition}
\newtheorem{lemma}[theorem]{Lemma}

\newenvironment{proof}[1][Proof]{\noindent\textbf{#1.} }{\hfill$\square$}

\newtheorem{definition}{Definition}[section]

\newcolumntype{L}[1]{>{\raggedright\let\newline\\arraybackslash\hspace{0pt}}m{#1}}
\newcolumntype{C}[1]{>{\centering\let\newline\\arraybackslash\hspace{0pt}}m{#1}}
\newcolumntype{R}[1]{>{\raggedleft\let\newline\\arraybackslash\hspace{0pt}}m{#1}}

\newenvironment{myquote}%
  {\list{}{\leftmargin=0.3in\rightmargin=0in}\item[]}%
  {\endlist}

\begin{document}

\begin{titlepage}
\title{\vspace{-3em}Going... going... wrong: a test of the level-\textit{k} (and cognitive hierarchy) models of bidding behaviour\thanks{I would especially like to thank my supervisor Miguel Ballester for his insightful suggestions and constant support. For useful comments and discussions, I would like to thank Johannes Abeler, Jesper Åkesson, Elizabeth Baldwin, Adam Brzezinski, Riccardo Camboni, Vince Crawford, Jasmin Droege, Caspar Jacobs, Nat Levine, Luke Milsom, Frieder Neunhoeffer, Rossa O'Keeffe-O'Donovan, Bernhard Kasberger, Paul Klemperer, Alex Teytelboym and Jasmine Theilgaard. I would also like to thank audiences at the Transatlantic Theory Workshop, SasCa PhD Conference, Oligo Workshop, and various seminars and workshops in Oxford. Finally, I am grateful to Darija Halatova and Tommaso Batistoni for coding assistance and to Exeter College and the OEP Fund for their generous financial support. \textit{Links}: \href{https://auctionsolver.herokuapp.com/}{\textcolor{blue}{algorithm}};
\href{https://www.socialscienceregistry.org/trials/8011}{\textcolor{blue}{pre-registration}};
\href{https://auction-experiment.herokuapp.com/demo}{\textcolor{blue}{experiment}}.}}
\author{Itzhak Rasooly\thanks{Department of Economics, University of Oxford.   } }
\date{November 1, 2021}
\maketitle
\vspace*{-0.4cm}
\begin{center}
\end{center}
\vspace*{-0.8cm}
\begin{abstract}
\noindent In this paper, we design and implement an experiment aimed at testing the level-$k$ model of auctions. We begin by asking which (simple) environments can best disentangle the level-$k$ model from its leading rival, Bayes-Nash equilibrium. We find two environments that are particularly suited to this purpose: an all-pay auction with uniformly distributed values, and a first-price auction with the possibility of cancelled bids. We then implement both of these environments in a virtual laboratory in order to see which theory can best explain observed bidding behaviour. We find that, when plausibly calibrated, the level-$k$ model substantially under-predicts the observed bids and is clearly out-performed by equilibrium. Moreover, attempting to fit the level-$k$ model to the observed data results in implausibly high estimated levels, which in turn bear no relation to the levels inferred from a game known to trigger level-$k$ reasoning. Finally, subjects almost never appeal to iterated reasoning when asked to explain how they bid. Overall, these findings suggest that, despite its notable success in predicting behaviour in other strategic settings, the level-$k$ model (and its close cousin cognitive hierarchy) cannot explain behaviour in auctions.

\noindent \\
\vspace{-0.75cm}
\vspace{0in}

\bigskip
\end{abstract}
\setcounter{page}{0}
\thispagestyle{empty}
\end{titlepage}
\pagebreak \newpage

\section{Introduction} \label{introduction}

Although the study of auctions remains dominated by equilibrium based approaches, in recent years more `behavioural' alternatives have emerged. Most prominent among these is \citeauthor{crawford2007}'s (\citeyear{crawford2007}) level-$k$ model, which posits that individuals iteratively best respond to a naive `first thought' about how one might bid in an auction. If accurate, the level-$k$ model would substantially revise our understanding of how auctions ought to designed \citep{crawford2009, mechanismdesign}; and would similarly revise how we think about individual reasoning in incomplete information settings. Understanding whether the level-$k$ model is accurate is thus a critical issue.

Surprisingly, however, the existing literature provides little insight into the predictive performance of level-$k$ and whether it can outperform the natural benchmark of Bayes-Nash equilibrium. There is some evidence that the level-$k$ model might be able to explain the ‘winner’s curse’ in common value auctions \citep{crawford2007, costa2015comment}; but also evidence that it might do so for the wrong reasons \citep{ivanov2010can}. In addition, there does not appear to be any serious comparison of level-$k$ and equilibrium in the rather simpler independent private setting (IPV) that forms the starting point for most auctions research. For example, while \cite{crawford2007} do compare level-$k$ and equilibrium using an IPV dataset, the auction structure that generated this dataset leads level-$k$ and equilibrium to make near identical predictions — making a formal comparison of the models almost impossible.\footnote{Indeed, 
\cite{crawford2007} restrict themselves to comparing level-$k$ with a rival behavioural theory, quantal response equilibrium.}\footnote{See also \cite{kirchkamp2011out} for a brief experimental comparison of level-$k$ and equilibrium, albeit in a paper not devoted to level-$k$ reasoning. Since the goal of their paper was not to assess the performance of level-$k$, it is perhaps not surprising that their experiments do not cleanly disentangle level-$k$ and equilibrium: for example, in their first experiment, the predictions of the models co-incide entirely. See also \cite{gillen2009identification}, \cite{an2017identification}, \cite{observational1} and \cite{observational2} for work on the level-$k$ model in the context of observational data. These papers do not attempt to test the accuracy of the level-$k$ model or compare it with equilibrium; and indeed doing so using observational data (when individual valuations are unknown) would obviously be much more challenging than doing so with an experimental approach.}


It is also difficult to learn very much about the relative performance of level-$k$ and equilibrium using existing experimental data sets since these have often been generated by first-price auctions with uniformly distributed valuations \citep{kagel1995auctions}. Unfortunately, the predictions of level-$k$ and equilibrium entirely coincide in such settings (for any $k \geq 1)$. Moreover, while it is easy to disentangle level-$k$ and equilibrium using `exotic’ value distributions, it is doubtful that such distributions would be understood by all subjects in any actual experiment. The challenge is therefore to find an experimental design that is both sufficiently simple for subjects to understand and yet dramatically disentangles the models.


We begin by proposing an environment that does exactly this: a discrete\footnote{We work with discrete models throughout both for the sake of realism (any actual auction must be discrete) but also because the level-$k$ model would be ill-defined in a continuous action setting (we elaborate more upon these points later on).} all-pay auction with uniformly distributed values. We first analyse the equilibrium predictions in this environment, proving that it possesses exactly one symmetric equilibrium. We thus extend a well known result from continuous auction theory to the discrete auction model -- and do so using completely novel (and elementary) arguments.\footnote{Our argument holds for arbitrary value distributions and can be straightforwardly modified to allow for factors like risk aversion. It also extends immediately to the first-price auction; indeed, we suspect that it could be extended to fairly arbitrary auction structures (although we have not formally verified this).} In the the course of our proof, we construct an algorithm\footnote{To aid future experimenters, we have implemented the algorithm and made it available online. See:  \textcolor{blue}{\url{https://auctionsolver.herokuapp.com/}}} that is guaranteed to identify this equilibrium, finding that it closely approximates the continuous equilibrium of traditional auction theory given the parameters chosen in our experiment. In this equilibrium, bids increase (quadratically) in valuations, rising to around half of valuations at the maximum possible valuation that players can draw.

We then conduct a similar analysis of the level-$k$ model but uncover a surprising contrast. Solving the level-$1$ player's optimisation problem reveals a corner solution: such types bid zero for all valuations. As a result, level-$2$ and level-$3$ types also bid close to zero; the former `overcutting' the level-$1$s by (almost always) bidding $1$ and the latter `overcutting' the level-$2$s by (almost always) bidding $2$. Clearly, this is very different from equilibrium; and so we have found an environment that dramatically separates the models. In addition, this environment leads the models to possess very different comparative statics. For example, in the level-$k$ model making the bid discretisation more coarse substantially inflates the predicted bids; whereas in the equilibrium model, it has very little effect.

Next, we conduct a similar exercise in the first-price auction, again assuming uniformly distributed values to make the environment experimentally implementable (and to force the two level-0 specifications proposed by \cite{crawford2007} to coincide -- we elaborate on this point later on). To break the equivalence of the level-$k$ and equilibrium models, we introduce the possibility that subject bids are ‘cancelled’ (in which case they lose the auction automatically). We then solve for the cancellation probability that maximises the ‘distance’ between the two theories. This turns out to separate the models in a similar way: level-$k$ bids are once again close to zero, whereas equilibrium bids are roughly a quarter of valuations. Thus, we once again manage to disentangle the models in an environment simple enough to be implemented in an experiment.

We then conduct an experiment designed to put these insights into practice. The basic idea of the experiment is straightforward. We create the aforementioned environments in a virtual laboratory, allowing us to see which of the models under consideration better predict the observed bids. However, the experiment also possesses various other features aimed at testing the predictive success of the level-$k$ model. For example, we elicit individuals’ levels using a variant on \cite{arad2012}’s 11-20 game; study whether behaviour varies with the bid discretisation in the manner predicted by level-$k$ theory; and finally ask individuals to explain why they bid in the way that they did to see if the level-$k$ model accurately characterises their conscious reasoning processes.

The experiment yields four central findings. First, if individual levels are restricted to the $1-3$ range (to capture an intuitive limit on the number of thinking steps that individuals can be expected to perform), then the level-$k$ predictions are far lower than the observed bids. In addition, the level-$k$ predictions are substantially less accurate than those of equilibrium. These conclusions hold across all the auction structures and treatments that we consider and under a variety of approaches to model testing and comparison. Moreover, this conclusion continues to hold if one allows individuals to vary across auction rounds but retains the requirement that individual levels are `plausible’ (i.e. in the 1--3 range). 

Second, if one fits the model without any such plausibility restrictions, one finds levels typically estimated in the 30--35 range. We view this almost as \textit{reductio ad absurdum} of the model since it conflicts with almost all prior experimental work on this issue \citep{crawford2013}, clashes with basic intuition about the number of thinking steps that individuals can plausibly be expected to conduct, and sits uneasily beside findings from linguistics that individuals struggle to even comprehend higher order statements whose order exceeds three \citep{arad2012}. Moreover, methods for penalising model flexibility (e.g. the Bayesian information criterion) generally do not support the inclusion of these very high levels.

Third, individuals’ levels as estimated from the auctions bear almost no resemblance to their levels as inferred from a game known to trigger level-$k$ reasoning (the 11-20 game). This is hard to square within the context of the level-$k$ model: for if levels reflect either cognitive sophistication or beliefs about the cognitive sophistication of one’s opponents, then one would expect to find a positive correlation between estimated levels across games. Conversely, the absence of positive correlation across any two games can be taken as \textit{prima facie} evidence that level-$k$ reasoning is not operative in at least one of the games under consideration.

Finally, we find that subjects are unlikely to cite iterated reasoning when asked to explain why they bid in the way that they did. In other words, subjects are unlikely to offer conjectures about the bidding strategies of their opponents -- and even less likely to ground these conjectures in what they take to be their opponents’ conjectures about their own bidding strategy. Admittedly, we do see evidence of iterated reasoning in the explanations given by at least one subject -- and this subject was (perhaps not coincidentally) likely to bid in line with the predictions of the level-$k$ model. However, such reports are very rare, suggesting that the level-$k$ model captures at best a small fraction of observed bidding behaviour.

We then conduct a series of extensions and robustness checks. For example, we consider the impact of incorporating risk aversion into the models, allowing subjects to best respond to a distribution over levels (as in \cite{camerer2004cognitive}), changing the $L_0$ specification, as well as conducting various tweaks to our empirical analyses (e.g. including subjects who submit dominated bids). We do not find that any of these exercises alter our substantive conclusions. 

The remainder of this article is structured as follows. Section \ref{separating} outlines an equilibrium and level-$k$ analysis of the discrete all-pay and first-price auctions. Section \ref{experimental_design} presents our experimental design. Section \ref{results} contains our empirical results and Section \ref{robustness_checks} our robustness checks. Finally Section \ref{Concluding_remarks} concludes with a discussion of why the level-$k$ model does so poorly in the auction setting despite its well documented successes in other domains of strategic behaviour.

\section{Separating level-$k$ from equilibrium} \label{separating}

\subsection{The all-pay auction}\label{theory}

 To begin, let us describe our first environment which sharply disentangles the level-\textit{k} and equilibrium models: a (discrete) all-pay auction with uniformly and independently distributed values. There are $n \geq 2$ bidders, each assumed to be risk neutral (we relax this assumption in Section \ref{robustness_checks}). As in our experiment, each player’s valuation is drawn from the set $\mathbb{X}=\{0,1..., x\}$ for some $x \in \mathbb{N^+}$; and the bids that they may submit are restricted to the very same set $\mathbb{X}$. Let $\beta   \colon \mathbb{X} \rightarrow \mathbb{X}$ denote a generic pure strategy (i.e. `bidding function'). We study the case of discrete values and bids since (i) this is inevitably the set-up in the experiment that follows, as indeed it must be in any auction experiment; (ii) if we had instead assumed that the bids were continuous, as is more standard in the literature, then the predictions of the level-\textit{k} model would be ill-defined.\footnote{In particular, the level-$2$ player would not possess a best response: for any number $\epsilon > 0$ that they might bid, there exists a feasible bid $\epsilon' < \epsilon$ that would yield them a higher expected pay-off.} 

Since this is an all-pay auction, individuals need to pay their bid no matter whether they win the object. However, they win the object if and only if their bid is strictly higher than all bids submitted by their opponents. Note that we assume that individuals cannot win the object in the event of a tie. While this is a slightly unusual assumption, it matches the tie-breaking rule that we adopt in the experiment. This tie-breaking rule, in turn, was chosen because it tremendously simplifies the equilibrium (and to a lesser extent the level-\textit{k}) analysis. 

To begin with, we consider what Bayes-Nash Equilibrium (henceforth ‘equilibrium’) predicts that individuals will do in this situation. In particular, we will focus on \textit{symmetric} equilibria. We do this for three reasons. First of all, it is unclear how individuals could be expected to coordinate on asymmetric equilibria without any communication in a one-shot game. Second, symmetric equilibria are analytically more tractable. Third, based on computational experiments reported in \citet{rasooly}, it does not seem as though discrete auction games like those described above possess asymmetric equilibria which are substantially different from their symmetric equilibrium. Therefore, one might hope that the restriction to symmetric equilibria is `almost' without loss of generality.

Unfortunately, since our auction is discrete we cannot directly apply results from ‘standard’ auction theory. We can, however, establish the following very useful result (which holds entirely independently of the distribution from which each player’s valuation is drawn).

\begin{proposition}\label{prop1}
The discrete all-pay auction has exactly one symmetric equilibrium.
\end{proposition}

Proposition \ref{prop1} contains two claims: that a symmetric equilibrium exists and that it is unique. Establishing the existence part is straightforward: since the game is finite, one can essentially just apply \cite{harsanyi1967}. However, establishing uniqueness is considerably more involved: see Appendix \ref{proofs} for the details. Importantly, the uniqueness proof is based on an \textit{algorithm} that explicitly constructs the (candidate) equilibrium (see Appendix \ref{algorithm} for examples of the algorithm in action). Thus, our result does not merely tell us that there is exactly one symmetric equilibrium, but also shows us exactly how to compute it.\footnote{This is invaluable from an experimental perspective since `brute force' approaches to equilibrium computation take several hours even for auctions with a handful of possible values and essentially intractable for auctions of the size considered in our experiment \citep{rasooly}. In contrast, our algorithm is able to compute the symmetric equilibria of our experimental games in less than a second.}

Having developed this algorithm, we now use it compute the symmetric equilibrium in the case of the distribution (uniform) used in our experiment. Figure \ref{fig1} displays the result for the case of $x =15$ and $n=3$: a dot represents a bid that is submitted with positive probability. As can be seen, the equilibrium is in (properly) mixed strategies; but it is closely approximated, quantitatively speaking, by the continuous equilibrium\footnote{See, for instance, \citet{klemperer1999} for a derivation.} of `textbook' theory:
\begin{equation}\label{all_pay_eq}
\beta(v) = \left( \frac{n-1}{n} \right)\frac{v^n}{x^{n-1}}
\end{equation}
As in the continuous model, equilibrium bids increase slowly when valuations are low but quickly when valuations are high. When values equal the maximum $x$, bids are roughly equal to the equilibrium bid in the first-price auction, i.e. $x(n - 1)/n$.

\begin{figure}[h!]
    \centering
    \caption{Equilibrium in the all-pay auction}
    \vspace{-1em}
    \includegraphics[width=14cm]{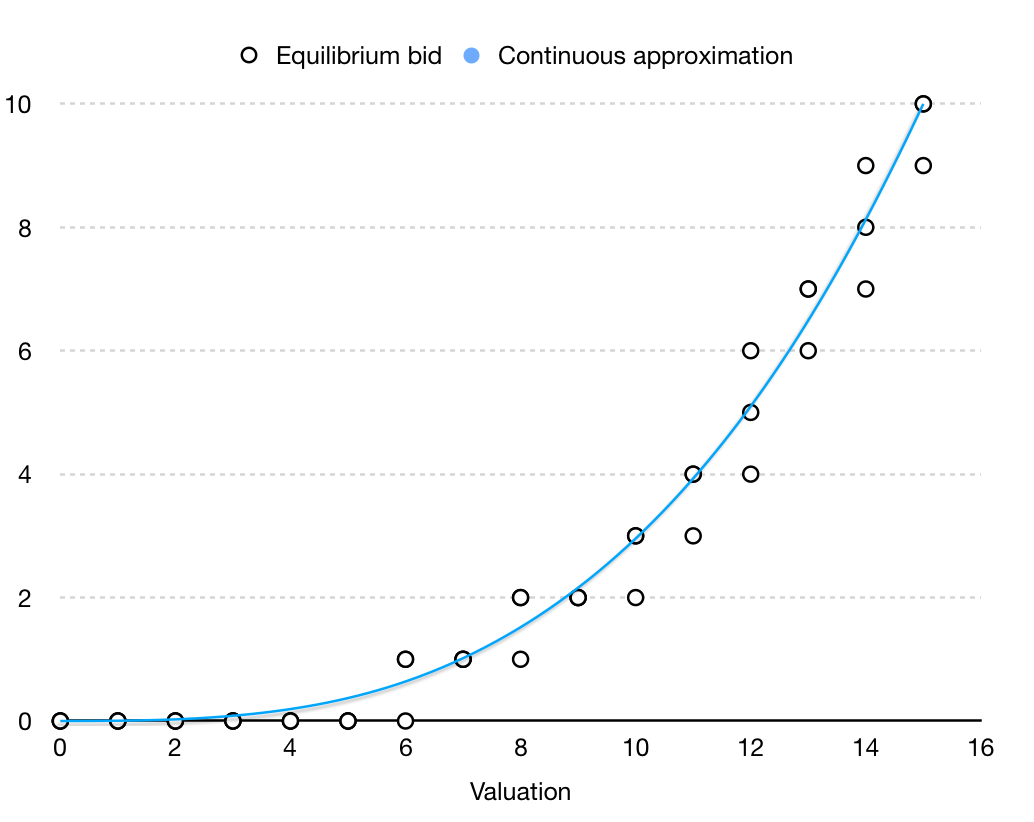}
    \label{fig1}
\end{figure}

Having discussed the predictions of equilibrium (our main benchmark), we now turn to the predictions of the level-\textit{k} model. As formulated by \citet{crawford2007}, the model begins with a `naive' type of player (the level-$0$) who either bids their valuation (the `truthful' specification) or either bids any element of the strategy space $\mathbb{X}$ with equal probability (the `random' specification). This is meant to capture a player's `first thought' about how their opponents will play the game. Higher types are then defined inductively: for $k \geq 1$, a level-$k$ player submits the bid that maximises their expected payoff given that their opponents are all level $k-1$ (see \citet{crawford2007} for interpretation and elaboration). In the following, we will write $\beta^k \colon \mathbb{X} \rightarrow \mathbb{X}$ to denote a level-$k$ player's bidding strategy.

We now derive the behaviour of the types. Assume, as in our experiment, that values are uniformly distributed on $\mathbb{X}$ (retreating from the generality of Proposition \ref{prop1}). Given the `truthful' specification ($b = v$), level-$0$ bids are then uniform on $\mathbb{X}$. Note, however, that this is also the `random' specification. We thus see that, quite aside from making the experimental rules easier for subjects to grasp, imposing uniform values ensures that the two level-$0$ specifications lead to identical predictions, a great advantage when it comes to testing the model.

Given that level-$0$ bids are uniform, the level-$1$ player solves the problem
\begin{equation}
\max_{b \in \mathbb{X}} \quad \pi(v, b) \equiv v\mathbb{P\text{(win}}|b) - b = v\left(\frac{b}{x+1}\right)^{n-1} - b.
\end{equation}
This leads to a corner solution: bidding 0 for all values $v \in \mathbb{X}$. To see why this is true, note that the benefit to such a player of increasing their bid by 1 is the gain in probability $\Delta \mathbb{P}$ multiplied by their valuation $v$. If $n = 2$, then $\Delta \mathbb{P} = 1/(x + 1)$; so the benefit of bidding one more is $v/(x + 1)$. Meanwhile, the cost of bidding one more is simply $1$. Since $v/(x + 1) < 1$ for all $v \in \mathbb{X}$, the benefit of increasing one’s bid is always lower than the cost -- which means that the optimal bid is 0. Moreover, if it is optimal to bid 0 when $n=2$, it must be optimal for any $n \geq 2$: adding more bidders decreases the expected payoff from submitting any positive bid, while leaving the pay-off from bidding 0 fixed.

It is then straightforward to derive the strategies of higher levels. For now, assume that a level-\textit{k} player breaks ties by choosing the lowest optimal bid (while this assumption is natural since it allows us to avoid dominated bids, we relax it in Section \ref{robustness_checks}). Given this assumption, a level-2 player sets $\beta(v) = 1$ for every value $v \geq 2$ (`over-cutting' their level-1 opponents); and sets $\beta(v) = 0$ otherwise. Similarly, a level-$3$ player sets $\beta(v) = 2$ for all $v \geq 3$; and sets $\beta(v) = 0$ otherwise. Thus, for all `reasonably low’ levels, a level-\textit{k} type player bids
\begin{equation}
\beta_{LK}(v) = 
\begin{cases} k-1 &\mbox{if } v \geq k \\
0 & \mbox{otherwise} \end{cases}
\end{equation}

The strategy $\beta_{LK}$ describes level-$k$ behaviour at `reasonably low' levels. At higher levels, however, this simple pattern breaks down. The strategy $\beta_{LK}$ calls for a level-$k$ player with a value of $k-1$ to bid $0$, thereby earning a (certain) payoff of $0$. There is another option, however: bid $1$ in the hope of winning when one's opponents all bid $0$. As $k$ rises, this second option becomes more attractive; and eventually it pays to deviate from $\beta_{LK}$. It is easy to check, however, that level-$k$ players must bid $\beta_{LK}$ provided that
\begin{equation}\label{inequality}
k \leq (x+1)^\frac{n-1}{n} + 1;
\end{equation}
an inequality that gives concrete meaning to the phrase `reasonably low' in the previous paragraph.

Finally, we argue that level-$k$ can never coincide with equilibrium; with the implication that level-$k$ bidding functions must cycle as $k \rightarrow \infty$. To prove this, note that, by Proposition 4 of \cite{rasooly}, the discrete all-pay auction with uniform values does not possess a symmetric equilibrium in pure strategies. Of course, it does have a symmetric equilibrium (see our Proposition \ref{prop1}), so this must be in mixed strategies. Now, given our assumption about level-\textit{k} tie breaking, it is clear that a level-\textit{k} player could never choose a mixed strategy. As a result, for any $k \in \mathbb{N}^+$, level-\textit{k} can never coincide with the symmetric equilibrium. Moreover, this means that the level-\textit{k} predictions must cycle. To see this, suppose that they did not cycle. Since the set of bidding functions is finite, there would need to exist some $k \in \mathbb{N}$ such that $\beta^k$ coincides with $\beta^{k'}$ for all $k' > k$. (Otherwise, a strategy would need to re-appear at some stage -- which would necessarily trigger a cycle.) But then $\beta^k$ would be a symmetric, pure strategy equilibrium, contradicting \cite{rasooly}'s Proposition 4.


The following proposition summarises the findings of the previous paragraphs.

\begin{proposition}\label{prop2}
For every level $k$ that satisfies (\ref{inequality}) and every value $v \in \mathbb{X}$, a level-$k$ bidder sets $\beta^k(v)=\beta_{LK}(v)$. Moreover, there is no $k \in \mathbb{N}$ at which $\beta^k$ coincides with equilibrium; with the implication that $\beta^k$ cycles as $k \rightarrow \infty$.
\end{proposition}

The figure below plots both the equilibrium predictions (or technically the continuous approximation) and level-\textit{k} predictions (for levels 1--3). As can be seen, the auction structure has done an excellent job disentangling the predictions of the models. First, given any reasonable calibration of the levels, the level-\textit{k} model predicts bids that are substantially lower than those predicted by equilibrium. Second, while the equilibrium bidding function is generally increasing, the level-\textit{k} bidding functions are flat at all but one value, providing a second contrast which one can test. Finally, our non-convergence result means that the theories remain separated at arbitrarily high levels. It follows then, that we have found a powerful way to disentangle the models, and to do so without abandoning the uniform values specification which is so useful when implementing any actual experiment.

\begin{figure}[H]
    \centering
    \caption{Comparing level-$k$ with equilibrium}
    \vspace{-0.3em}
    \hspace{-0.9em}
    \includegraphics[width=14cm]{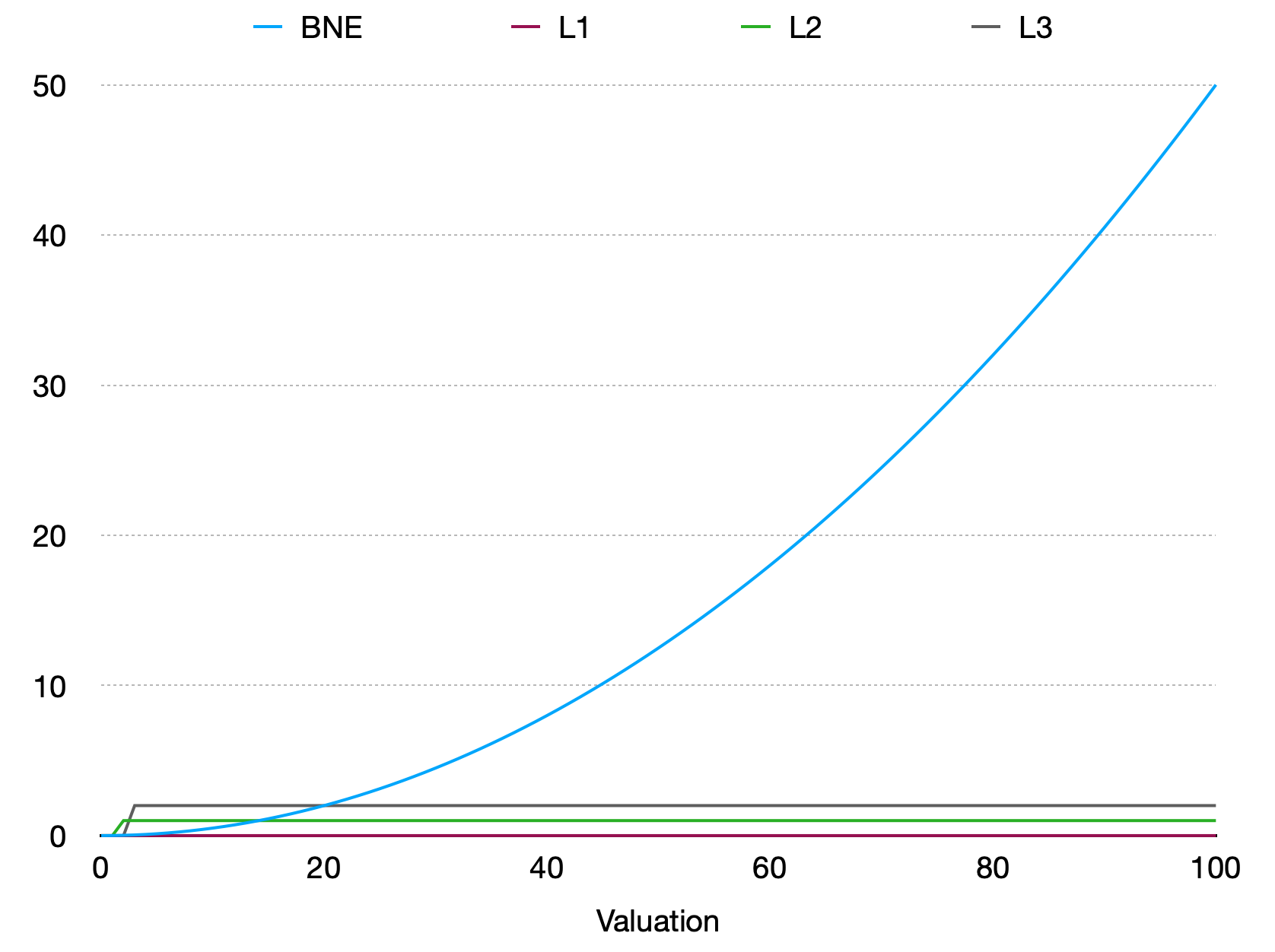}
    \label{fig2}
\end{figure}

\subsection{The first-price auction}\label{first_price_auction}

It is clearly desirable to find a second auction structure which sharply separates the models -- that way, we can check the relative performance of the models in two different strategic settings. To this end, we now examine the possibility of doing this within the first-price auction. As noted by \cite{crawford2007}, the level-\textit{k} and equilibrium models coincide (for every $k \geq 1$) in the first-price auction with uniformly distributed values. To break this equivalence, we now suppose that each player's bid only goes through with probability $p \in (0, 1)$; otherwise their bid is cancelled and they cannot win the auction. We now consider what probability $p$ maximally disentangles the theories.

Formally, the model structure is much as before: for example, values and bids are restricted to the same set of consecutive integers $\mathbb{X}$. However, there are two differences. First, since this is now a first-price auction, a player pays their bid if and only if they win the object. Second, as just mentioned, each player's bid is now cancelled with probability $1 - p \in (0, 1)$; and a player wins if and only if their (non-cancelled) bid is strictly higher than all the non-cancelled bids submitted by their opponents.

As before, we begin by demonstrating that a symmetric equilibrium exists and is unique\footnote{While our uniqueness result still holds under arbitrary value distributions, it now requires the restriction that $p < 1$; the role of this assumption is to ensure that all bids have positive probability of winning and thus to rule out equilibria involving dominated bids. Since this restriction holds in our experiment, this technical point need not concern us here.}:

\begin{proposition}
If $p \in (0, 1)$, then the discrete first-price auction has exactly one symmetric equilibrium.
\end{proposition}

The arguments used to establish this result are very similar to those used to prove Proposition 1. In particular, the argument is based on the very same algorithm that constructs the equilibrium. As before, using this algorithm reveals that the equilibrium bids are well approximated by the continuous equilibrium, which is now (under uniform values)\footnote{Equation (\ref{formula}) may be derived by noting the formal equivalence of our model with a model with an uncertain number of bidders. One can then apply Theorem 1 of \citet{harstad1990}. Alternately, one can note that our model is also equivalent (at least as far as undominated equilibria are concerned) with a model in which each player's value distribution has a mass of probability on zero. One can then use the `standard formula' (see, e.g., \citet{krishna2009})
$$ \beta(v) = v - \int_0^v \frac{G(y)}{G(v)} dy, $$
and insert $ G(v) = \left(1 - p + pv/x \right)^{n-1}$ (recalling that $G(v)$ is the probability that the maximum valuation is $v$ or lower).}
\begin{equation}\label{formula}
\beta(v) = \left(\frac{n-1}{n}\right)v - \frac{x(1 - p)}{np}\left[ 1 - \left(\frac{1 - p}{1 - p + p(v/x)}\right)^{n-1} \right]
\end{equation}

As can be seen from Figure 2, introducing the possibility that bids are cancelled makes the bidding less aggressive. If $p = 1$, i.e. bids are never cancelled, our expression reduces to $\beta(v) = v(n-1)/n$. In other words, we recover the standard formula for equilibrium with uniformly distributed values. Moreover, it is easy to check that for any $v\in \mathbb{X}$, $\beta(v)$ is strictly decreasing in $p$. Finally it can be shown (using L'H\^{o}pital's rule) that $\beta(v) \rightarrow 0$ as $p \rightarrow 0$ for all $v \in \mathbb{X}$. This is also not surprising: if one is almost certain to win the auction anyway (in the event that one's own bid is not cancelled), then one may as well bid $b \approx 0$.

We now turn to the predictions of the level-\textit{k} model. Given that level-$0$ bids are either cancelled or uniform on $\mathbb{X}$, a level-$1$ player solves the problem
\begin{equation}\label{objective}
\max_{b \in \mathbb{X}} \quad  (v-b)\mathbb{P\text{(win}}|b) = (v - b)\left(1 - p + \frac{pb}{x+1}\right)^{n-1}
\end{equation}
Since the objective function in (\ref{objective}) is single peaked, it is obvious that the integer that maximises this function can be obtained by finding the $b \in [0, x]$ that maximises the function and rounding up or down. In other words, we have a good continuous approximation of the exact level-$1$ bidding function. To simplify formulae, we will therefore proceed as if level-$1$ types could choose any $b \in [0, x]$ (although we will always use the exact model predictions when analysing the data from our experiment). Routine optimisation subject to the inequality constraint then reveals that
\begin{equation}
\beta^1(v) = \max\left\{\left( \frac{n - 1}{n} \right)v - \left(\frac{1 - p}{p}\right)\frac{x}{n}, \hspace{0.3em}0\right\}
\end{equation}
As in equilibrium, level-$1$ bids are decreasing in the probability $p$. Moreover, for all $p \leq 1/n$, the level-$1$ player sets $\beta(v) = 0$ for all $v\in \mathbb{X}$, i.e. we recover the same corner solution as in the all-pay auction. Intuitively, if $p$ is quite low, then there is a large chance that the player’s opponents’ bids are all cancelled -- so the player may as well just bid $0$ and receive their entire valuation whenever they `get lucky'. 

Our goal is to disentangle the models. To this end, we define the distance between the theories as
\begin{equation}\label{distance}
d(p) \equiv \int_0^x \left| \beta(v) - \beta^{1}(v) \right| dv.
\end{equation}
Of course, the distance $d(p)$ is simply the total area between the (approximated) level-$1$ and equilibrium bidding functions. We now study the probability $p^*$ that maximises this distance. Recall from above that the level-$1$ player bids $\beta(v) = 0$ for all $v \in \mathbb{X}$ provided that $p \leq 1/n$. Meanwhile, every equilibrium bid $\beta(v)$ is strictly decreasing in $p$ (moving closer to 0). From this, we see that we could never maximise (\ref{distance}) by choosing some $p \in [0, 1/n)$. The reason is that increasing $p$ (by a sufficiently small $\epsilon > 0$) would increase every equilibrium bid $\beta(v)$ while leaving the level-\textit{k} bids unchanged, thus increasing the distance between the two theories. In other words, we have the following easy result.
\begin{proposition}\label{prop3}
$p^* \geq 1/n$.
\end{proposition}

Unfortunately, $p^*$ cannot be solved analytically. Nonetheless, we can compute the optimal $p$ numerically given various values of $n$ (see Appendix \ref{optimising}). We find, for example, that $p^* \approx 0.54$ if $n = 2$, $p^* \approx 0.34$ if $n = 3$ and $p^* \approx 0.26$ if $n = 4$. Overall, then, choosing $p = 1/n$ seems like a good approximation of the optimal probability (the loss in terms of distance is small). This has the further advantage of making the probabilities easier for subjects to grasp than the exact solutions would be; so we ultimately opt for $p = 1/n$ in our experiment.

Assuming that $p=1/n$, it is then straightforward to derive the exact level-\textit{k} bidding strategies. As noted, the level-$1$ player sets $\beta^1(v) = 0$ for all $v \in \mathbb{X}$. As a result, the level-$2$ player bids $1$ (when their value is high) or bids $0$ (when their value is low). Similarly, the level-$3$ player bids $2$ for high values and $0$ otherwise. This simple pattern characterises the behaviour of level-$k$ bidders when $k$ is `reasonably low' relative to the number of valuations $x + 1$. More formally, we can say the following:

\begin{proposition}\label{prop4}
    Fix a set of types $\mathbb{K} = \{1, 2, ... \bar{k} \} $, let $p = 1/n$, and define
    $$v^*(k) = \frac{k-1}{1-\left(1-p\right)^{n-1}}$$
    Then if $x$ is sufficiently large, every level-$k$ type in $\mathbb{K}$ bids $\beta^k(v) = k - 1$ for $v > v^*(k)$ and $\beta^k(v) = 0$ otherwise.
\end{proposition}

Given the parameters ultimately chosen for our experiment ($n=2$, $p=1/2$, $x = 100$), level-$k$ behaviour once again cycles at very high levels. To illustrate this, Figure \ref{cycling} plots the maximum bid submitted by the first $100$ levels (the full level-$k$ bidding functions are available in the supplementary materials). As can be seen, the model produces cycles of length $29$ (so level-$7$ coincides with level-$36$, which in turn coincides with level $65$, etc.) Moreover, far from cycling locally around the equilibrium predictions, level-$k$ cycles are quite dramatic: the maximum bid ranges from 0 to 20. This is all the more surprising once we realise that all bids above 50 are strictly dominated (by bidding zero), so the highest level-$k$ bid traces out around 40\% of the non-dominated strategy space.

\vspace{-0.6em}

\begin{figure}[H]
    \centering
    \caption{Level-$k$ cycling}
    \vspace{-0.5em}
    \hspace{1em}
    \includegraphics[width=14cm]{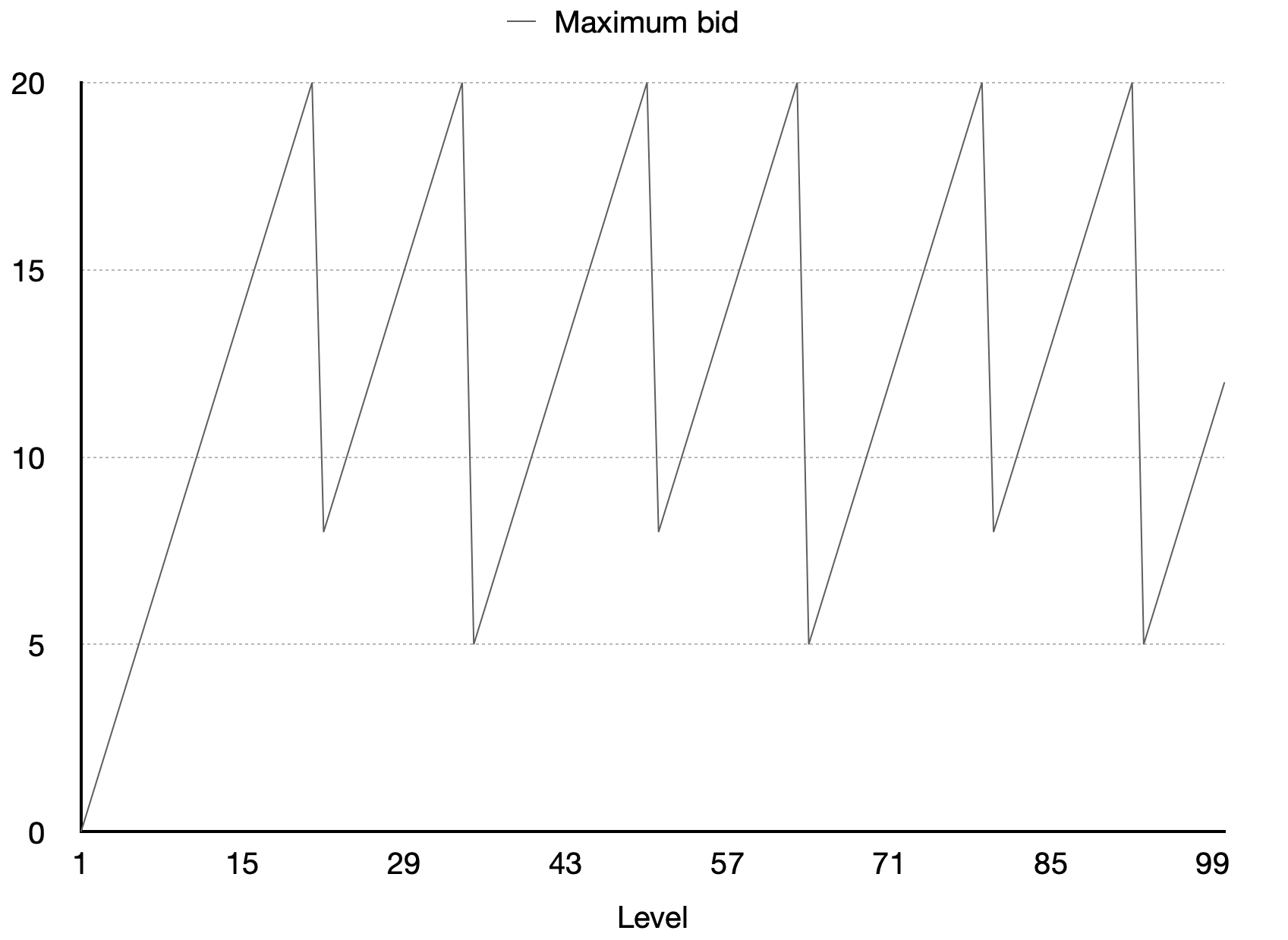}
    \label{cycling}
\end{figure}

Finally, Figure \ref{fig4} compares the models in the case of $n = 2$ and $p = 1/2$. As can be seen, we have once again managed to disentangle the theories -- and in a dramatic way. First, given any plausible calibration of the levels -- recall that these are supposed to be in the 1--3 region -- level-\textit{k} bids are significantly lower than equilibrium bids. Second, while equilibrium bids are increasing in values, level-\textit{k} bids are almost always non-increasing. Finally, the non-convergence of level-$k$ to equilibrium means that the models remain separated even if one were to ascribe implausibly high levels to the experimental participants.

\begin{figure}[H]
    \centering
    \caption{Comparing level-$k$ with equilibrium}
    \vspace{-0.5em}
    \includegraphics[width=14cm]{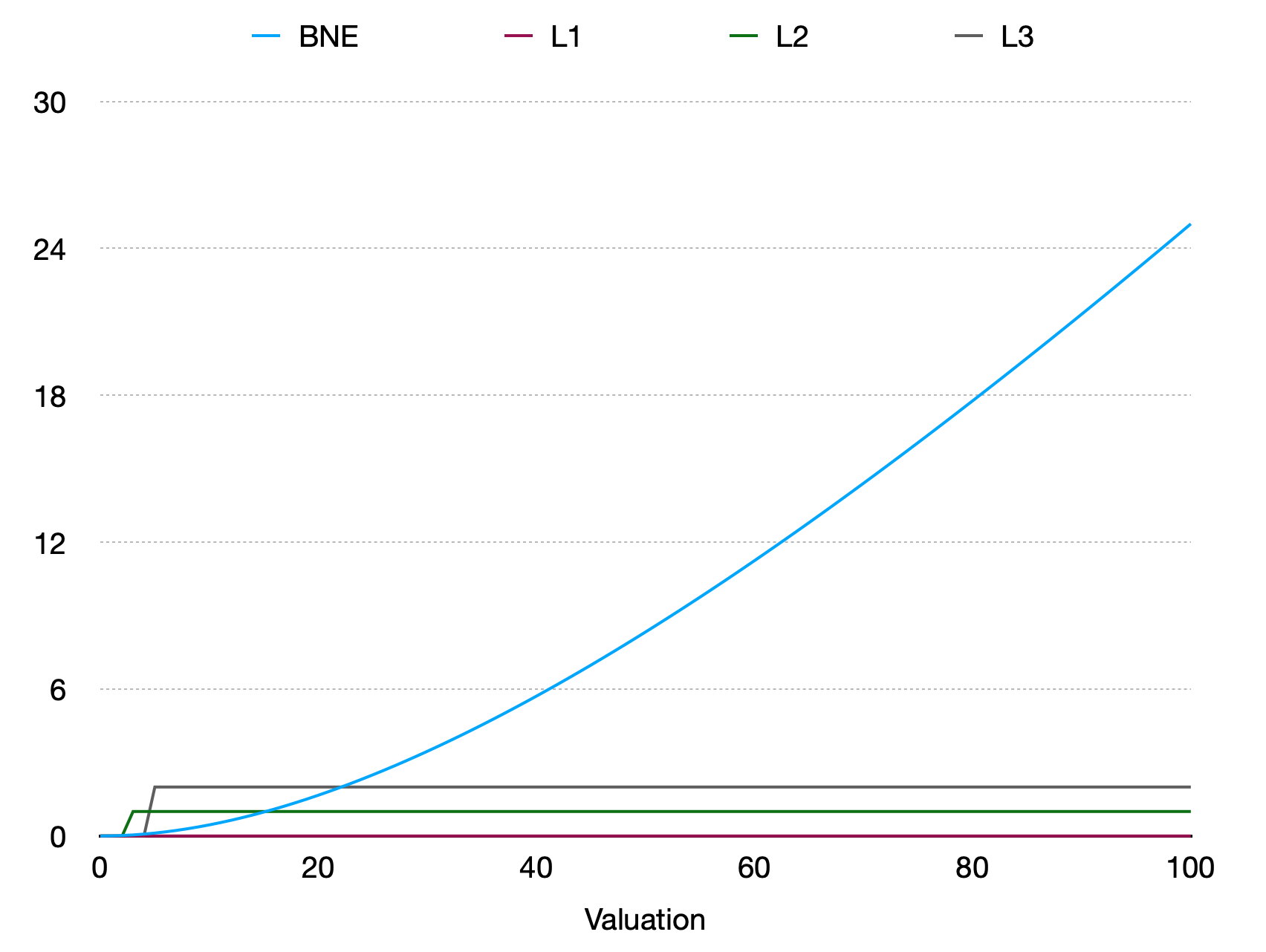}
    \label{fig4}
\end{figure}

\section{Experimental design} \label{experimental_design}

To test the insights of the previous section, we conducted an experiment with 84 participants in collaboration with the Nuffield Centre for Experimental Social Sciences. All participants were students at the University of Oxford and no participant took part in more than one session. The experiment was conducted online and written in oTree \citep{otree}.\footnote{The full experimental instructions can be viewed here: \color{blue}{\url{https://auction-experiment.herokuapp.com/}}}

The essential idea of the experiment is straightforward -- to implement the two aforementioned auction structures in a virtual laboratory and ask subjects how much they wanted to bid. It is then a relatively straightforward matter to check which of the two theories under consideration can better predict (and rationalise) the observed bidding behaviour.

More precisely, each subject in the experiment participated in a first-price auction (as described previously) with $n=2$, $p=1/2$ and an all-pay auction (as described in the previous section) with $n=2$. Following the auctions, subjects played a modified version of \citet{arad2012}’s 11-20 game that is proposed in \cite{alaoui2016}. The purpose of this game was to allow for a rough calibration of each participant’s level.\footnote{We use \cite{alaoui2016}’s modification of the original game since it prevents cycling (level-$k$ behavior coincides for all $k \geq 9$). As \cite{alaoui2016} note, this is then useful if one wishes to calibrate individual levels without the potentially controversial assumption that each individual’s level is the smallest possible level that rationalises their observed behaviour.} Finally, we measured each subject's risk aversion using the `bomb' risk elicitation task \citep{crosetto2013}.

Before participating in either of the auctions, subjects were given extensive quizzes to check their understanding of the rules. They were not permitted to proceed past the quiz until they had answered a full set of questions correctly. We view this as an important feature of the study: subjects are unlikely to make meaningful choices if they cannot even understand the rules of the game that they are playing. 

Each subject bid in two rounds of each of the auction types.\footnote{They were also asked to explain the general considerations determining how they chose to bid after every second round.} Moreover, we did not inform subjects of the bids submitted by their opponents in previous rounds. We made these two choices -- only two rounds coupled with very little feedback -- since the level-$k$ model is supposed to be a model of initial play, as opposed to a model of the outcomes of a long-run learning process \citep{crawford2007}.\footnote{Indeed, it is for this reason that \citet{crawford2007} only test their model using data from the first five rounds of the experiments they examine.} In other words, we sought to test the level-$k$ model within its intended domain of applicability.

Subjects were informed (and regularly reminded) that they would never bid against the same opponent twice across the rounds of an auction. That is, we used the `perfect stranger matching' protocol. We opted for this protocol since the two theories under consideration are static theories -- in the sense that they treat the game as a one-shot interaction -- and we did not want to complicate the analysis by introducing learning or reputational considerations. 

Bidding was incentivised through a procedure similar to that employed in \citet{ozbay2007}. In every round, individuals were asked how much they would want to bid given ten possible valuations that they might have. They were informed that they should bid carefully since one of these valuations would be their actual valuation and they would be committed to bidding the amount they had entered. In other words, one in ten bids were selected to ‘count’. We opted for this procedure, as opposed to incentivising every bid, since it allowed us to collect more data given a fixed number of rounds, thereby further reducing feedback or learning effects (and making the environment more hospitable to the level-$k$ model).

Finally, we randomly allocated the participants into two treatments: an ‘integer bid’ treatment (to which 2/3 of participants are allocated), and a treatment in which bids were constrained to be multiples of 5. The idea here was to check whether varying the bid treatment has the strong effects predicted by level-$k$ theory or the relatively weak effects predicted by the equilibrium model. In the level-$k$ model, increasing the discretisation by a factor of 5 increases the predicted bids by roughly a factor of 5, assuming that levels are sufficiently low for individual behavior to be characterised by Propositions \ref{prop2} and \ref{prop4} (see the supplementary materials for an exact description of the effects of the bid discretisation on the model’s predictions). In contrast, varying the bid discretisation does not have large quantitative effects on equilibrium bids: for example in the first-price auction the average equilibrium bid increases by just 4.7\% (again, see the supplementary materials).

The experiment was pre-registered along with a full analysis plan.\footnote{See \textcolor{blue}{\url{https://www.socialscienceregistry.org/trials/8011}}.} The experiment also received Ethics Approval from both the University of Oxford’s and Centre for Experimental Social Sciences' institutional review boards.

\section{Results} \label{results}

\subsection{Pooling the data}\label{pooling}

We begin by asking which theory best explains the overall patterns observed in the bidding data. To this end, we pool all the valuations and bids, compute the mean observed bid for every valuation and plot this against the corresponding predictions of the two models. For the moment, we assume that levels cannot exceed three\footnote{We impose this constraint because all of the evidence of which we are aware suggests that insofar as individuals use level-$k$ reasoning, it rarely goes beyond the third level. Indeed, this is what we find when examining the responses of our subjects in the 11-20 game: out of those who mentioned level-$k$ reasoning when asked to explain their choices, not one made a choice consistent with more than three levels of reasoning.} and drop the handful of subjects who submitted at least one dominated bid in one of the auctions.

Figure \ref{t1fp} displays the results for the case of the first price auction with integer bids (see also Figure \ref{t2fp} for the treatment in which bids are multiples of 5). As found in many prior studies of first-price auctions, we observe substantial overbidding relative to the (risk-neutral) equilibrium benchmark \citep{kagel1995auctions}. Indeed, the average bid is around 46\% of the average value, whereas equilibrium predicts that it should be around 20\%. However, while the equilibrium model struggles to match the data, it is clear that the level-$k$ model performs considerably worse, predicting bids that are at least an order of magnitude lower than those commonly observed. As a result, it is clear that equilibrium provides a substantially better (albeit highly imperfect) description of the average bidding function observed in the experimental data.

\begin{figure}[H]
    \centering
    \caption{Predictions and data in the first-price auction}
    \vspace{-0.2em}
    \includegraphics[width=14cm]{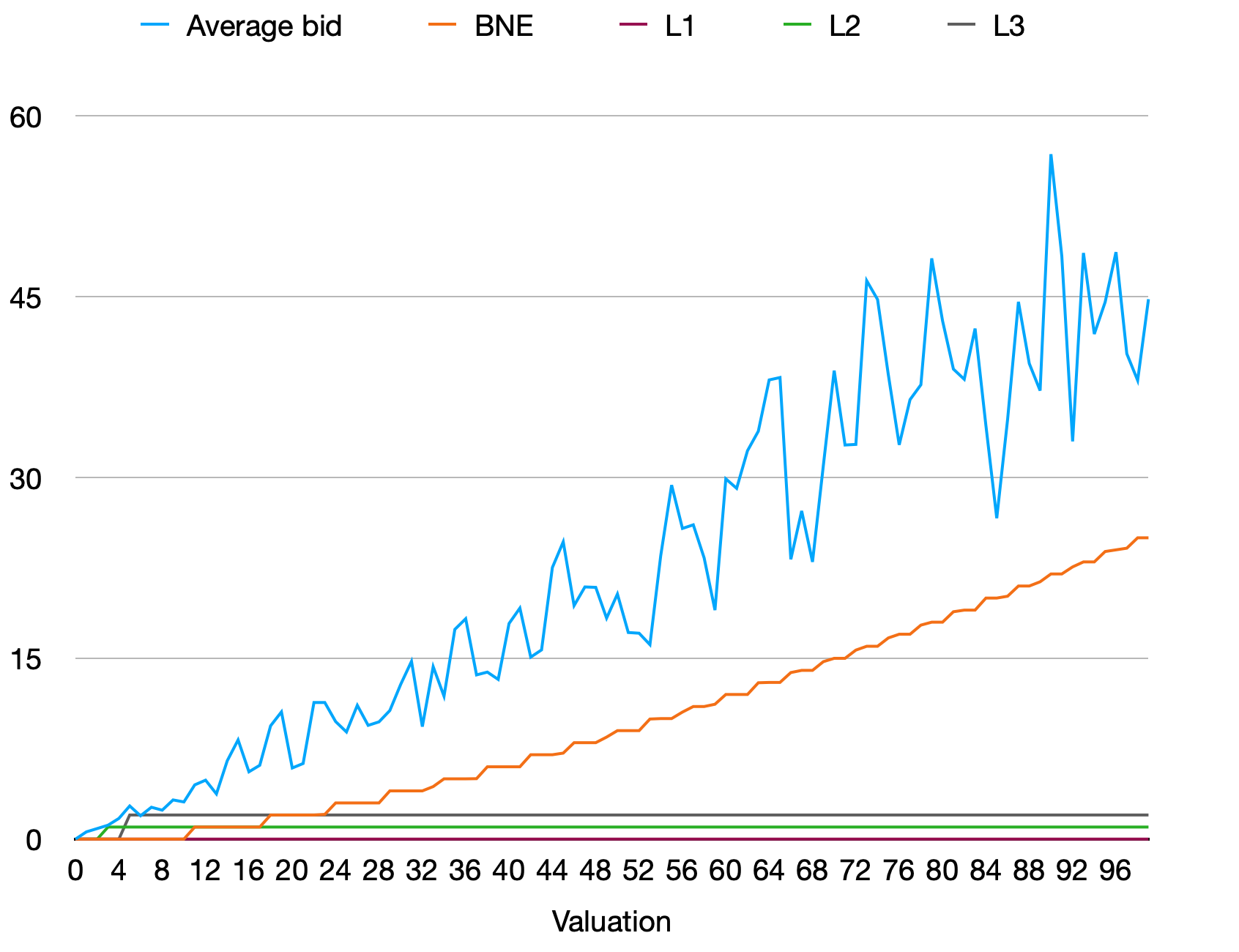}
    \label{t1fp}
\end{figure}

We then conduct a similar analysis of the all-pay auction. As shown in Figure \ref{t1ap}, equilibrium does a somewhat better job of tracking mean bids in this environment (see also Figure \ref{t2ap} for the case where bids are multiples of 5). However, over-bidding relative to equilibrium is again observed on average at low to medium values.\footnote{The overbidding that we observe is consistent with results from the only other experimental study of IPV all-pay auctions of which we are aware \citep{noussair2006behavior}. Unfortunately, however, we were unable to obtain the data from this experiment and so were unable to use it to complement our analyses.} As before, the level-$k$ model predicts bids that are orders of magnitude below those that are observed on average, at least if one imposes the restriction that levels cannot exceed 3. Thus, our initial examination of the data suggests a very clear ranking of the two models under consideration.

\begin{figure}[H]
    \centering
    \caption{Predictions and data in the all-pay auction}
    \vspace{-0.2em}
    \includegraphics[width=14cm]{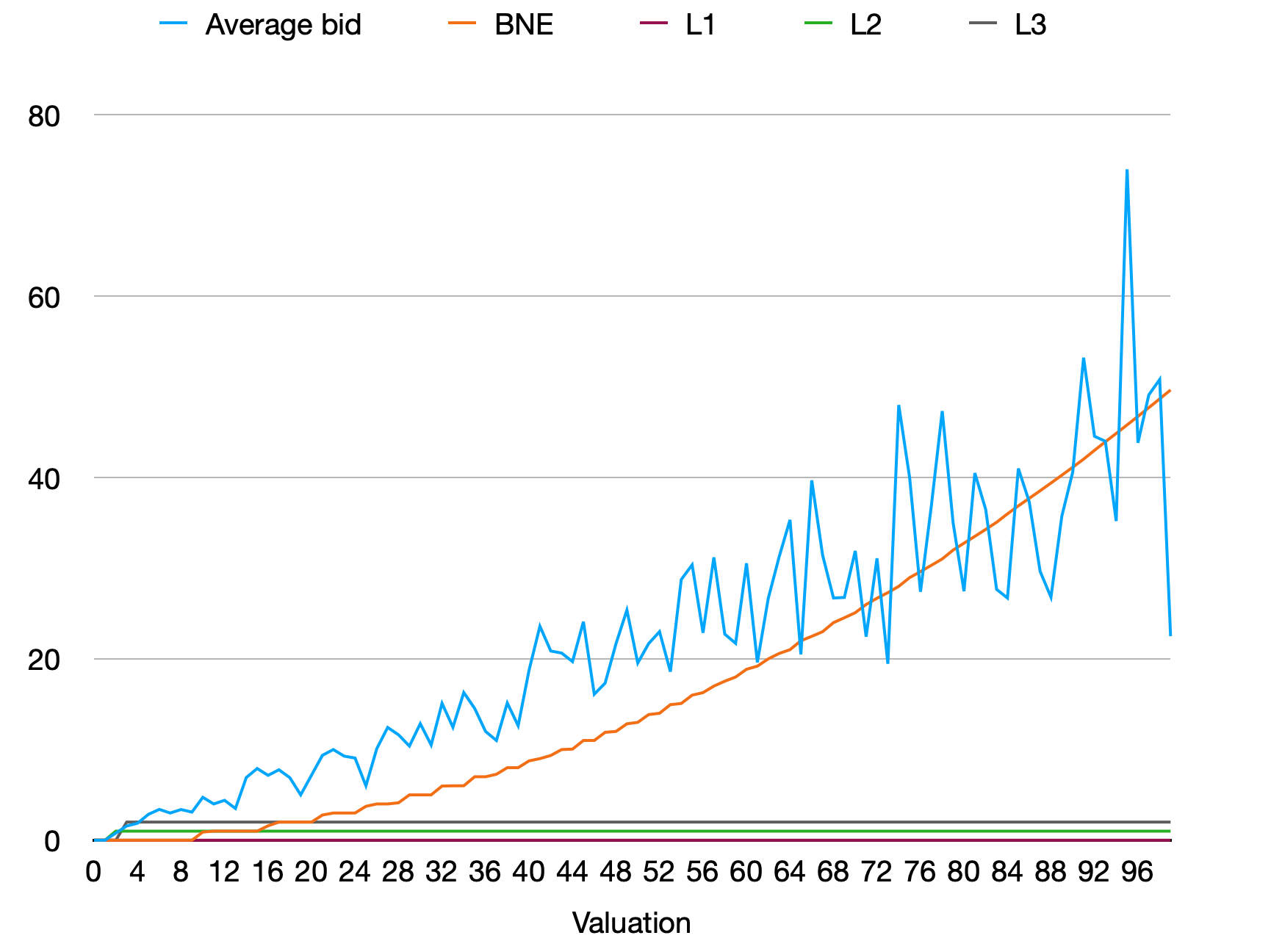}
    \label{t1ap}
\end{figure}

To verify these conclusions more formally, we also compute the root-mean-square prediction errors of both of the theories. In the case of equilibrium, this is straightforward since the model is `parameter free'. In the case of level-$k$, we obtain a prediction by assigning each datapoint the level from the $1-3$ range that minimises the model's prediction error (this is a generous assumption that will tend to overstate the model's predictive performance). Table \ref{errors} displays the results. As can be seen, both models make size-able prediction errors. Consistent with the graphical evidence, however, the prediction errors of level-$k$ are substantially greater than those of equilibrium.

\begin{table}[H]
\begin{threeparttable}
\caption{Prediction errors}
\label{errors}
\begin{tabular}{l|cccc}
\hline
\textbf{} & T1 FP & T1 AP & T2 FP & T2 AP \\ \hline
Equilibrium       & 15.1  & 13.3  & 14.1  & 14.3  \\
Level-$k$ & 21.3  & 19.8  & 21.4  & 24.0  \\ \hline
\end{tabular}
\begin{tablenotes}
\footnotesize
\item \hspace{-0.45em}\textit{Notes}: T1 and T2 denote the integer bid and multiples of five treatments; FP and AP abbreviate ‘first-price’ and ‘all-pay’.
\end{tablenotes}
\end{threeparttable}
\end{table}

Next, we examine the effect of varying the bid discretisation. Recall that the level-$k$ model predicts a large effect from changing the discretisation from integers to multiples of 5: bids should increase by roughly a factor of 5. In contrast, the equilibrium model predicts much  smaller effects. Given these very different predictions, we now conduct a comparison of the integer bid and multiples of 5 bid treatments. 

As a preliminary, we confirm that the treatments are balanced, confirming that the randomisation worked as expected (see Table \ref{balance}). Next, we inspect the average bids across the two treatments and auction structures (see Table \ref{average_bids}). As can be seen, changing the bid discretisation makes essentially no difference to the first-price auction (in both cases, the average bid is about 46\% of the average valuation). In the all-pay auction, restricting bids to multiples of five does appear to slightly increase the average bids (which rise from 45\% to 50\% of the average valuation), but clearly not in the dramatic way predicted by the level-$k$ model. In neither auction, then, do the data support the surprising comparative static prediction of level-$k$ theory.

\subsection{Structural estimates}

In light of the evidence above, it is clear that the level-$k$ model cannot count for observed average behaviour, at least when calibrated with a plausible distribution of types. However, this still leaves at least two issues unresolved. First, even if level-$k$ cannot describe the average bidder, it is still possible that it accounts for a substantial minority of the observed behaviour. Second, it remains to be seen whether the model might be rescued by allowing for a greater number of possible types. We now address both of these issues.

To do this, we follow the literature in estimating a mixture of types model (see, e.g., \cite{stahl1994experimental}, \cite{bosch2002one} and \cite{crawford2007}). To obtain stochastic versions of the relevant theories, we assume that subjects submit bids with probabilities that are proportional to a normal distribution whose mean is the predicted bid that they are `supposed' to submit (according to the relevant theory). More formally, let $\mathbb{K} \subset \mathbb{N}^+$ denote the set of types, indexing equilibrium by $k=0$ and level-$k$ types in the obvious way. Given that a player $i$ has value $v_i \in \mathbb{X}$ and type $k \in \mathbb{K}$, we assume that the chance they submit a bid $b_i \in \mathbb{X}$ is given by:

\begin{equation}\label{normal_noise}
\mathbb{P}(b_i|v_i, k) \propto \exp \left[ -\frac{1}{2}\left(\frac{b_i - b_p}{\sigma_k}\right)^2 \right]
\end{equation}

where $b_p$ is the bid predicted by their type (given their valuation), $\sigma_k$ is a noise parameter that controls the variance of their bids, and the constant of proportionality is chosen to ensure that the probabilities sum to $1$.\footnote{We also experimented with a model containing logit errors. However, these had the unfortunate tendency of preferring level-$1$ to all higher levels even though level-$1$ predictions were the furthest from the observed bids (recall that level-$1$ types bid zero regardless of their valuation).} This produces a single peaked distribution whose mode is the predicted bid $b_p$. As $\sigma_k \rightarrow \infty$, this distribution converges to uniform noise; and as $\sigma_k \rightarrow 0$, it converges to a degenerate distribution which places all probability mass on the predicted bid $b_p$.

For now, assume that each individual's type is fixed (though we relax this in many of our analyses). Then the probability that an individual submits a vector of (twenty) bids $\boldsymbol{b_i}$ given their vector of values $\boldsymbol{v_i}$ is (assuming independence of errors):
\begin{equation}\label{bid_vector}
\mathbb{P}(\boldsymbol{b_i}|\boldsymbol{v_i}, k) = \prod_{i=1}^{20}\mathbb{P}(b_i|v_i, k)
\end{equation}
We assume that the individual we observe is drawn randomly from the population; and write $p_k \in [0, 1]$ to denote the population proportion of type $k$ (obviously, $\sum_k p_k = 1)$. By the law of total probability, the unconditional probability of their bids is
\begin{equation}\label{uncond}
\mathbb{P}(\boldsymbol{b_i}|\boldsymbol{v_i}) = \sum_{k \in \mathbb{K}} \mathbb{P}(\boldsymbol{b_i}|\boldsymbol{v_i}, k)p_k
\end{equation}
and so the `likelihood' of all the observed decisions is (combining Equations \ref{normal_noise} - \ref{uncond})
\begin{equation}\label{likelihood}
L = \prod_{i=1}^n \mathbb{P}(\boldsymbol{b_i}|\boldsymbol{v_i}) \propto 
\prod_{i=1}^n
\left(
\sum_{k \in \mathbb{K}}
\left[
\prod_{i=1}^{20}
\left(
\exp\left[ -\frac{1}{2}\left(\frac{b_i - b_p}{\sigma_k}\right)^2 \right]
\right)
\right]
p_k
\right)
\end{equation}
In every estimation exercise, we compute the parameters $\boldsymbol{p} = (p_k)_{k \in \mathbb{K}}$ and $\boldsymbol{\sigma} = (\sigma_k)_{k \in \mathbb{K}}$ that maximise the likelihood (\ref{likelihood}).  We then calculate standard errors with the jack-knife. While we consider the case where individual levels are fixed (as above), we also consider the case where individual levels are permitted to vary by round -- although we do insist throughout that levels are fixed within a round.\footnote{Recall that all bids within a round must be made simultaneously, so there is no possibility for learning within rounds. In contrast, individuals could conceivably learn between rounds -- although the scope for such learning is limited by the near-total absence of feedback.}

As a preliminary exercise, we compare an equilibrium only model with the level-$k$ model in which levels are restricted from $1-3$. To avoid introducing too great an imbalance in the number of parameters between the two respective models, we restrict both to a single noise parameter $\sigma$. Table \ref{horse_race} displays the results. As one might have expected from Figures \ref{t1fp} - \ref{t1ap}, equilibrium does a better job than any of the level-$k$ models at explaining the data, attaining a higher log-likelihood in all specifications. This advantage becomes even greater if we punish the level-$k$ models for their extra flexibility; for example, using the Bayesian information criterion. This simply reinforces the conclusion that we reached earlier: equilibrium does a better job than the level-$k$ model at accounting for the overall patterns in the data.

We now examine the extent to which the level-$k$ model might be salvaged by adding higher levels. To this end, we now estimate a series of models, gradually adding levels until we reach the point where the model cycles. Three main points emerge from this exercise.\footnote{Since the resulting tables are too large to be contained within the appendix, they have been made available \href{https://docs.google.com/spreadsheets/d/1qNrerSNlzhqOSkyISZL29tgxcjxqW_3wO3refseew3o/edit}{\textcolor{blue}{here}}.}
First, while the level-$k$ model can achieve a higher likelihood than equilibrium once the full suite of levels is added, it can do this only by placing most probability on very high levels: for example, in the first price auction, most probability mass is placed on level-$32$. As argued previously, is doubtful whether we can plausibly ascribe so many thinking steps to participants. Second, the level-$k$ model fits the data with a very unusual distribution, placing some probability mass on level-1, none on intermediate levels, and the rest on the very high levels just discussed. Of course, the idea that individuals do either one thinking step or 32 thinking steps might appear to be rather strange. Finally, the inclusion of these very high levels is almost never recommended by the Bayes-Information Criterion. For all these reasons, it seems difficult to rescue the level-$k$ model through a strategy of adding a very large number of levels.

Although a pure level-$k$ model appears to be outperformed by equilibrium, a final question we investigate is whether a hybrid model might be able to outperform both. To this end, we now re-estimate the model for all auction structures, allowing for both equilibrium and level-$k$ types and no longer restricting the model to a single noise parameter $\sigma$. We also return to our previous assumption that levels are in the plausible range ($1$-$3$). Table \ref{hybrid} displays the results. As can be seen, most probability mass is placed on equilibrium -- suggesting again that the level-$k$ model fits the data relatively poorly. However, there is still some probability mass placed on level-$k$ types. Thus, the results from the hybrid model are consistent with the possibility that the level-$k$ model characterises the behaviour of a small minority of subjects, even if it struggles to capture the general patterns in the bidding data.

\subsection{Correlating levels}

We now examine whether individual levels as estimated from the two auctions bear any resemblance to the levels as inferred from a game known to trigger level-$k$ reasoning: a variant on the 11-20 game proposed by \cite{alaoui2016}. To this end, we calculate the level $k$ and precision $\sigma_k$ that maximises the likelihood of each individual's bids (i.e. Equation \ref{bid_vector}). Thus, we now assign each individual a definite level, as opposed to a distribution of levels; and no longer allow for equilibrium types. After doing this for both the first-price and all-pay auction, we then compute each individual's level as inferred from the 11-20 game, discarding those whose implied level exceeds 4.\footnote{We do this partly because such levels are at odds with prior experimental evidence but also because not one of such subjects invoke level-$k$ reasoning when asked to explain their choice.} Finally, we compute the correlation between the levels inferred from each auction and the level as inferred from the 11-20 game.

Figure \ref{correlate} plots the levels estimated from the first-price and all-pay auctions against the levels as inferred from the 11-20 game. As can be seen, there is no discernible relationship between the levels; and indeed the correlations are estimated to be $0.00$ and $-0.09$ in the case of the first-price and all-pay auctions respectively.\footnote{This result reflects those of \citet{georganas2015}, who find no correlation between estimated individual levels across two families of games (albeit in a non-auction setting).} From the perspective of level-$k$ theory, this is rather unexpected. For if levels reflect either cognitive sophistication or beliefs about the cognitive sophistication of one’s opponents -- the two leading interpretations of individual levels -- then one would expect to find a positive correlation between individual levels across games. Conversely, the absence of such a positive correlation suggests that individuals are not using level-$k$ reasoning in at least one of the games.

\subsection{Subject reports}\label{reasoning}

Finally, we examine the explanations given by participants as to why they chose to bid in the way that they did.\footnote{Historically, there seems to have been a widespread belief in economics that one cannot learn anything from individuals from asking them why they behaved in the way they did (see e.g. Becker ??). As far as we can see, however, this belief was not adopted on the basis of any kind of empirical evidence. More to the point, it is obviously false: for example, macroeconomists appear to learn a great deal from unemployment statistics, which in turn are almost always based on individual self-reports \citep{thaler2015}.}  In particular, we consider whether they appealed to iterated reasoning.\footnote{A strict definition of iterated reasoning would be reports of the form: ``I bid in way X because I expected my opponent to bid in way Y; and I expected my opponent to bid in way Y because I thought they would think I would bid in way Z." In other words, this would need to involve (1) a conjecture about the bidding strategies followed by the player's opponents (2) an attempt to ground this conjecture using a conjecture about the conjectures of the opponents about the player's bidding strategy. On a weaker definition, we would simply require some kind of conjecture about the opponents' bidding strategy (along with an attempt to optimise given this conjecture). Since the level-$k$ model allows for level-$1$ players, we used this weaker (and more permissive) definition when categorising the reports.} Assuming that subjects’ conscious reasoning was indeed captured by the level-$k$ model and further assuming that subjects did not for some reason wish to lie to their experimenter about their reasoning, one would expect to see iterated reasoning cited in their explanations. Conversely, one can take the absence of such reports as evidence that the level-$k$ model did not in fact characterise their conscious reasoning.

In their explanations, the majority of subjects (around 55\%) explicitly indicated an awareness of the basic bidding trade-off: bidding more increases one’s probability of winning, but decreases one’s pay-off if one wins the auction. However, the vast majority of subjects did not appear to use iterated reasoning. Indeed, we categorise responses as invoking iterated reasoning in just two cases.\footnote{All categorisations were independently checked by another researcher who was not involved in the study. The reader is invited to examine the responses themselves which are provided (along with all other non-demographic data generated by the experiment) in the supplementary materials.}

To be clear, we are not claiming that the level-$k$ model did not characterise the reasoning given by any of the subjects. For example, one subject (who also bid in line with the model) wrote that
\begin{myquote}
Since my opponent has a 50\% chance of their bid being cancelled, it's better to keep bids low in hopes theirs [sic] gets cancelled but mine does not.  I raised the bid to more than 1 because in the event my opponent is using the same strategy and bids very low, I still have a chance of winning the auction.
\end{myquote}

Thus, this subject not only considered the level-$1$ strategy (bidding zero) but then went beyond it -- a clear case of level-$k$ reasoning.\footnote{While the bids submitted by the subjects did not exactly conform to any of the levels they were very close, with the subjects bidding a mix of 2's and 3's at relatively low valuations.} What is evident, however, is that reasoning was provided by only a small minority of subjects, a finding that is also in line with the choice data outlined in the previous section.

\section{Robustness checks}\label{robustness_checks}

We now report on the results of a variety of extensions and robustness checks.

\textbf{Risk aversion}. While the previous analysis assumed that individuals are risk-neutral, we now consider whether the level-$k$ model might be salvaged by allowing individuals to be risk-averse.\footnote{Of course, expected utility maximisers should be approximately risk-neutral over such small stakes \citep{arrow1970, rabin}. Thus, the `risk aversion' discussed here should really be thought of as capturing utility curvature within a reference dependent model \citep{bleichrodt2019}.} To begin, consider the all-pay auction. As noted in Section \ref{theory}, level-1 players (if assumed to be risk-neutral) are predicted to bid 0 (for all values), thereby earning a certain payoff of zero. If they submitted a positive bid, then they would be exposed to risk: either their pay-off would be positive (if they win the auction) or negative (if they lose the auction). Now, given that level-1 players choose to bid 0 when they are risk neutral, it is obvious that they must also choose to bid 0 if they are assumed to be risk averse. Thus, introducing risk aversion makes no difference to the model’s predictions -- and also (for this reason) makes no difference to the predictions for high levels (provided that they are sufficiently low to be characterised by bidding strategy $\beta_{LK}$).

While risk aversion cannot improve the model’s performance in the all-pay auction, matters are more complicated in the first-price auction. In general, holding the bidding strategies of one’s opponents fixed, introducing risk aversion must (weakly) increase one’s optimal bid. In our setting, this means that introducing risk aversion can make positive bidding by a level-1 type optimal, thereby bringing the model’s predictions closer in line with the observed data. Note, however, that introducing risk aversion \textit{also} makes equilibrium bids more aggressive \citep{holt1980competitive}, which might be expected to increase the predicted accuracy of our benchmark. It is therefore unclear at the onset whether introducing risk aversion improves the relative performance of the level-$k$ model.

To study this issue more carefully, we assume that every individual $i$ has constant relative risk aversion (i.e. $u_i(x) = x^\alpha_i$ where $x_i$ is the individual's pay-off and $\alpha_i > 0$ determines their level of risk aversion.) We then estimate every individual’s parameter $\alpha_i$ using their choice in the ``bomb'' risk elicitation task \citep{crosetto2013}. We then re-calculate the level-$k$ model’s predictions for every player using their estimated parameter $\alpha_i$ (along with their estimated level as in Section \ref{pooling}). We also re-calculate the equilibrium predictions on the assumption that every individual has the average estimated risk aversion parameter of the subject pool ($\alpha =0.72$).\footnote{The code can be viewed here: \textcolor{blue}{\url{https://github.com/Itzhak95/risk_aversion/tree/master}}. The $\alpha =0.72$ is obtained after dropping the two outliers who collected $61$ or more out of the $62$ boxes. If such subjects are included, we obtain $\alpha  = 1.8$ (risk-seeking), in which case allowing non-linear utility is unable to improve the performance of the level-$k$ model.}


Figure \ref{risk_integers} displays the results for the main treatment (see also \ref{risk_fives} for the case where bids are multiples of five). As can be seen, both the level-$k$ and equilibrium predictions are now somewhat improved and considerably closer to the actual data. Overall, however, it is clear that equilibrium still outperforms level-$k$ dramatically -- a conclusion confirmed by the prediction error rates reported in Table \ref{robustness}.

\begin{figure}[h!]
    \centering
    \caption{Risk aversion in the first-price auction}
    \vspace{-0.2em}
    \includegraphics[width=14cm]{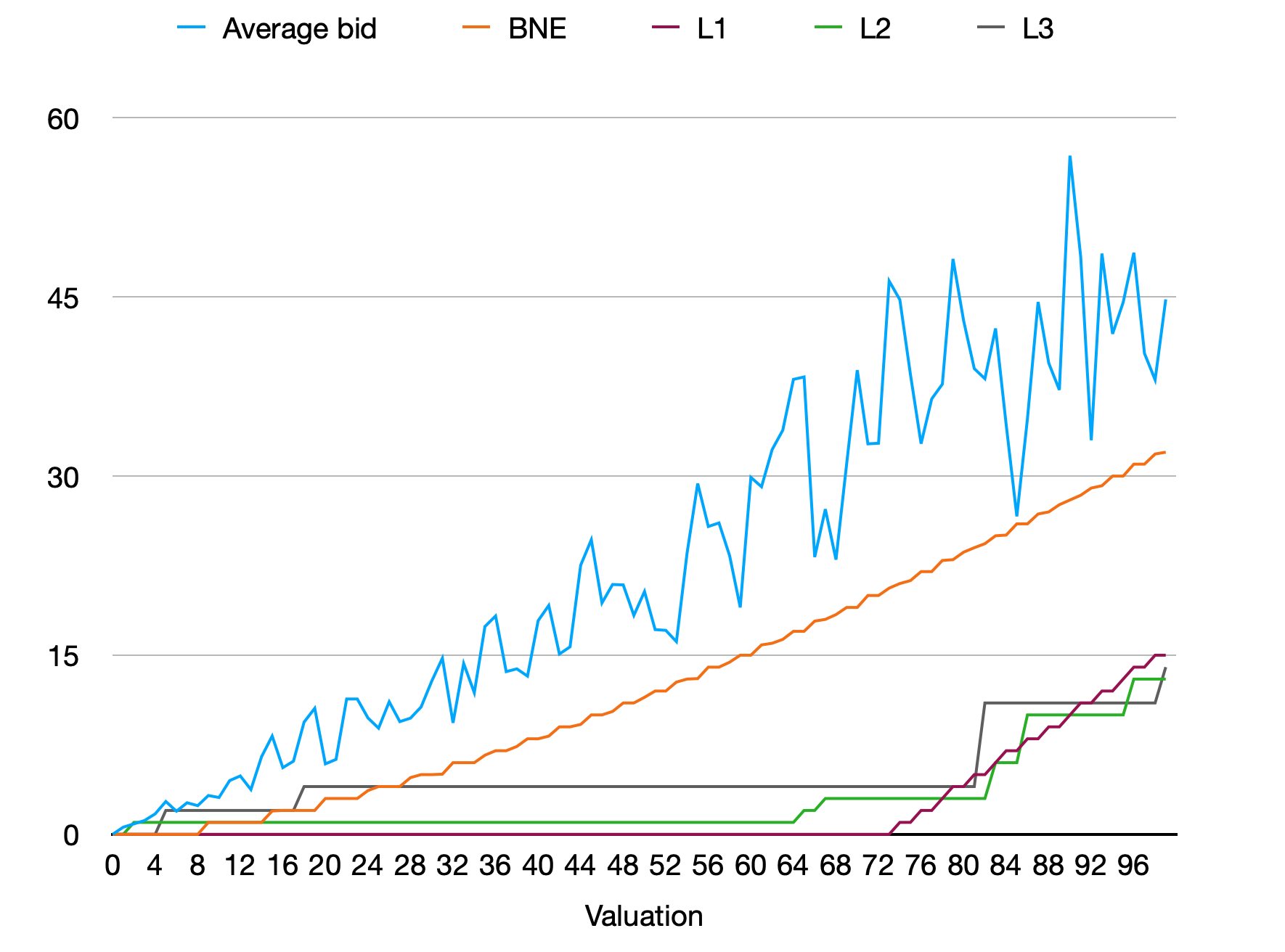}
    \label{risk_integers}
\end{figure}

\textbf{Cognitive hierarchy.} We now ask if any conclusions might change if individuals are assumed to best respond to a \textit{distribution} over levels, as in \citet{camerer2004cognitive}’s cognitive hierarchy model. To this end, we assume (following \citet{camerer2004cognitive}) that levels are Poisson distributed with mean $\tau$. Using standard maximum likelihood methods, we then estimate $\tau$ for the participants in our experiments using their choices in the 11-20 game, finding that $\tau \approx 1.9$. Finally, we compute the predictive errors of the cognitive hierarchy model and compare those with the associated errors for equilibrium. As shown in Table \ref{robustness}, moving from level-$k$ to cognitive hierarchy makes essentially no difference to the prediction errors. This is not a surprise: obviously, best responding to a mix of levels between 1 and 3 yields (in our context) very similar optimal bids to best responding to level-2. 

\textbf{Including dominated bids.} As noted earlier, our main analysis excluded the small number of subjects who submitted at least one dominated bid on the grounds that such subjects may have misunderstood the rules of the auction. However, excluding such subjects is a debatable choice; and we re-ran all analyses after including such subjects. This did not change any of our substantive conclusions (see Table \ref{robustness}).

\textbf{Changing the tie-breaking rule.} In the previous analysis, we assumed that level-$k$ players always chose the lowest optimal bid if multiple bids were optimal. While this allowed the model to avoid predicting that individuals choose dominated strategies, it is natural to ask whether any of our conclusions hinge on this assumption. To investigate this, we now make the opposite assumption -- that level-$k$ players always choose the highest bid out of the bids that they believe are optimal.\footnote{Obviously, this does not exhaust the universe of possible tie-breaking rules: one could also assume that players sometimes choose the highest optimal bid and sometimes choose the lowest optimal bid.} Re-running our main analyses, we see that this slightly worsens the model's fit with the data (although the change is small). This is not a surprise: the main change is that the model now predicts some dominated bids, and such bids are very rare in the data.

\textbf{Dropping the second round.} As noted in Section \ref{experimental_design}, the level-$k$ model is supposed to be a model of initial play. In order to test the model in its natural domain, we therefore blunted subject feedback and restricted the number of rounds to two. Nonetheless, one might worry that even two rounds is too many to count as `initial play'. To test this possibility, we re-ran the analyses after dropping the second round. While this does  somewhat improve level-$k$'s performance, the model remains clearly outperformed by equilibrium (see Table \ref{robustness_checks}).

\textbf{Changing the $L_0$ specification}. The level-$k$ auction model, as formulated by \cite{crawford2007}, assumes that level-$0$ players either bid their value (the `truthful' specification) or randomise uniformly over $\mathbb{X}$. Since the predictions of such a model seem strongly disconfirmed by the data, one might wonder whether the model might be improved by choosing an alternative level-0 specification. While our data do not allow us to absolutely rule this possibility out, there are various reasons to think that this is not likely to be a fruitful theoretical exercise.

First, as stressed by \cite{crawford2013}, one must be careful to ensure that any L0 specification represents a `strategically naive' assessment of how others will play the game. This is most obviously important if one wishes to connect the model to its standard interpretation as a model of iterated reasoning given some naive but initially plausible starting point. In addition, it is vital to restrict the range of possible level-$0$ strategies if one wishes the model to make anything like definite predictions in strategic situations; and to be clearly disentangled from rival models. For instance, if an analyst were allowed to specify that level-$0$ play coincided with the symmetric equilibrium, it would be impossible to disentangle the equilibrium from the predictions of the level-$k$ model: for every $k \geq 1$, the level-$k$ prediction would simply be the symmetric equilibrium strategy.

Thus, if one did want to change the $L_0$ specification, one would need to make sure that it represents a psychologically plausible way in which a player might naively expect others to play. This brings us to our second point, however: there simply do not appear to be a large number of such $L_0$ specifications. Other than the specifications proposed by \cite{crawford2007}, we are only able to think of one alternative: bidding uniformly between 0 and one’s value $v$ (modifying their `random' specification so as to remove dominated bids). However, opting for this L0 specification would render the model unable to match even the general trend of bidding in the most experimentally studied of all auction formats -- the first-price auction with independent, uniform values (and no possibility of cancelled bids).\footnote{In this setting, a large number of auction experiments report overbidding relative to the risk-neutral equilibrium $\beta(v) = (n-1)v/n$ (see \cite{kagel1995auctions} for a review). However, it is easy to see the level-$0$ specification suggested here would predict bids that are substantially lower than the risk-neutral equilibrium. As a result, it would result in level-$k$ predictions that are far too low (at least, given a plausible calibration of the levels).

Note also that other apparently plausible $L_0$ specifications, such as thinking that one's opponents bid half their value, are excluded on the grounds that $L_0$ reasoning is supposed to be non-strategic (whereas bidding a fraction of one’s value quite clearly reflects an attempt to strategically balance one’s chance of winning with one’s payoff if one does win). Moreover, such specifications imply that level-$k$ bidders should submit a large mass of bids just above the maximum value bid by $L_0$ (for instance $L_1$ players in the all-pay auction should bid $50$ for all $v \geq 51$). We do not observe such masses in this experiment nor in the other auctions experiments of which we are aware.}

Finally, we note that the level-$k$ predictions as formulated in Section \ref{separating} are robust to a very wide range of level-0 specifications. For example, in the all-pay auction, level-$k$ bids will be characterised by Proposition \ref{prop2} for \textit{any} $L_0$ specification that satisfies
\begin{equation}\label{permissive}
    F(b) \leq  \bigg(\frac{b}{x+1}\bigg)^{\frac{1}{n-1}} 
\end{equation}
where $F(b)$ is the probability that the level-$0$ type bids $b$ or lower (i.e. the CDF of $b$). In particular, this means that Proposition ($\ref{prop4}$) will continue to hold provided that the CDF $F(b)$ that is bounded by the uniform CDF from above. Similarly, in the first-price auction, level-$1$ players will bid 0 (as in Proposition $\ref{prop4}$) for any level-$0$ specification that satisfies
\begin{equation}
    F(b) \leq \left( \frac{n-1}{n} \right)\left(\left(\frac{x}{x-b}\right)^{n-1}-1 \right)
\end{equation}
for all $b \in \mathbb{X}$.\footnote{This inequality comes from checking that a player with $v = x$ would rather bid $0$ than any $b \in \mathbb{X}$.} Moreover, while these inequalities are sufficient for our level-$k$ characterisation to hold, they are absolutely not necessary. For example, in the all-pay auction, the level-1 bidding strategy turns out to be a best response to the empirical bid distribution observed in the experiment, even though that distribution does not satisfy inequality (\ref{permissive}).

In light of the reasons above, we think that it is unlikely that one can find a level-0 specification that is simultaneously non-strategic, psychologically plausible and generates predictions that match the data in both this and other auction experiments. Ultimately though, this conclusion is somewhat more speculative than the central conclusion of this study: that the level-$k$ model as specified by \cite{crawford2007} cannot predict or rationalise the observed bidding behaviour.

\section{Concluding remarks} \label{Concluding_remarks}

In this paper, we have designed and implemented an experiment aimed at testing the level-$k$ model of auctions. Overall, the evidence would appear to strongly reject the model. When plausibly calibrated, the model produces bids that are at least an order of magnitude too low. Moreover, fitting the model without such a plausibility constraint results in estimated levels in the 30--35 range; and these levels in turn fail to bear any relation to levels inferred from the 11-20 game. In addition, changing the bid discretisation has little effect on the observed bids (in contrast to the model’s predictions); and subjects almost never cite iterated reasoning when asked to explain why they bid in the way they did. Thus, while one might be able to disagree with any of the pieces of evidence separately, it seems that the evidence collectively casts strong doubt on the model’s ability to predict or even rationalise the observed data. 

This finding might be seen as surprising given that the experimental environment was carefully chosen to be the level-$k$ model’s natural domain of applicability. Individual feedback was severely restricted and (following \cite{crawford2007}) the number of rounds was restricted to just two. However, the level-$k$ model fared poorly despite the initial play setting.

The findings also might be seen as surprising in light of evidence that the level-$k$ model very accurately captures facets of human behaviour in other settings, whether `beauty contests' \citep{nagel}, co-ordination games \citep{costa2009}, strategic communication \citep{crawford2003lying, cai2006} or zero sum betting \citep{brocas2014}. We thus close with some speculations as to why the model appears to do so badly at predicting bids in auctions despite its successes in predicting behaviour in other areas.

First, recall that the level-$k$ model requires a level-0 anchor, i.e. some sort of naive starting point that forms the basis of the iterated reasoning. In many settings, such an anchor is easy to provide. For example, one’s natural first thought (in a strategic communication setting) is that one’s opponents will tell the truth; and a natural first thought (in the context of guessing games) is that all possible guesses are equally likely. In contrast, it is not at all clear how the level-0 ought to be specified in the setting of auctions. \cite{crawford2007} select what appear to be the two best proposals -- that subjects either bid their valuation or randomise uniformly over the strategy space. However, since both of these proposals involve dominated bids, it is not obvious whether either represents a natural first thought for how one’s opponents will play. More to the point, it does not seem that there exists more plausible L0 specifications than those proposed by \cite{crawford2007} (see section \ref{robustness_checks} for discussion). Thus, the failure of the level-$k$ model in this context may well be down to the lack of a salient and psychologically plausible anchor.


Second, iterated reasoning is cognitively much more taxing in Bayesian games like auctions than the simpler settings to which the level-$k$ model is normally applied. To calculate one’s optimal action, one must combine one’s conjecture about one’s opponent’s strategy with their distribution over types to arrive at their distribution over actions. One must then solve an (often non-trivial) optimisation problem to find the optimal action given this distribution. Thus, even if a plausible level 0 anchor were available, it is unclear whether individuals would be able to correctly solve for the L1 strategy -- and even less clear whether they would be able to correctly identify the higher order strategies.

If this is right, it follows that the level-$k$ model should be applied only in situations which possess both an intuitively appealing `first thought’ along with easily computable best response dynamics. Rigorously testing whether these features are indeed necessary for the predictive success of the model -- along with identifying strategic settings which possess these apparently necessary features -- would seem to be an important task for future research.

\clearpage

\setlength{\bibhang}{0pt}
\bibliographystyle{apalike}
\bibliography{bibliography.bib}

\begin{thebibliography}{}

\bibitem[Alaoui and Penta, 2016]{alaoui2016}
Alaoui, L. and Penta, A. (2016).
\newblock Endogenous depth of reasoning.
\newblock {\em The Review of Economic Studies}, 83(4):1297--1333.

\bibitem[An, 2017]{an2017identification}
An, Y. (2017).
\newblock Identification of first-price auctions with non-equilibrium beliefs:
  A measurement error approach.
\newblock {\em Journal of econometrics}, 200(2):326--343.

\bibitem[Arad and Rubinstein, 2012]{arad2012}
Arad, A. and Rubinstein, A. (2012).
\newblock The 11-20 money request game: A level-\textit{k} reasoning study.
\newblock {\em American Economic Review}, 102(7):3561--73.

\bibitem[Arrow, 1970]{arrow1970}
Arrow, K.~J. (1970).
\newblock Essays in the theory of risk-bearing.
\newblock Technical report.

\bibitem[Bleichrodt et~al., 2019]{bleichrodt2019}
Bleichrodt, H., Doctor, J.~N., Gao, Y., Li, C., Meeker, D., and Wakker, P.~P.
  (2019).
\newblock Resolving rabin’s paradox.
\newblock {\em Journal of Risk and Uncertainty}, 59(3):239--260.

\bibitem[Bosch-Domenech et~al., 2002]{bosch2002one}
Bosch-Domenech, A., Montalvo, J.~G., Nagel, R., and Satorra, A. (2002).
\newblock One, two,(three), infinity,...: Newspaper and lab beauty-contest
  experiments.
\newblock {\em American Economic Review}, 92(5):1687--1701.

\bibitem[Brocas et~al., 2014]{brocas2014}
Brocas, I., Carrillo, J.~D., Wang, S.~W., and Camerer, C.~F. (2014).
\newblock Imperfect choice or imperfect attention? understanding strategic
  thinking in private information games.
\newblock {\em Review of Economic Studies}, 81(3):944--970.

\bibitem[Cai and Wang, 2006]{cai2006}
Cai, H. and Wang, J. T.-Y. (2006).
\newblock Overcommunication in strategic information transmission games.
\newblock {\em Games and Economic Behavior}, 56(1):7--36.

\bibitem[Camerer et~al., 2004]{camerer2004cognitive}
Camerer, C.~F., Ho, T.-H., and Chong, J.-K. (2004).
\newblock A cognitive hierarchy model of games.
\newblock {\em The Quarterly Journal of Economics}, 119(3):861--898.

\bibitem[Chen et~al., 2016]{otree}
Chen, D.~L., Schonger, M., and Wickens, C. (2016).
\newblock otree—an open-source platform for laboratory, online, and field
  experiments.
\newblock {\em Journal of Behavioral and Experimental Finance}, 9:88--97.

\bibitem[Cheng et~al., 2004]{cheng}
Cheng, S.-F., Reeves, D.~M., Vorobeychik, Y., and Wellman, M.~P. (2004).
\newblock Notes on equilibria in symmetric games.

\bibitem[Costa-Gomes et~al., 2009]{costa2009}
Costa-Gomes, M.~A., Crawford, V.~P., and Iriberri, N. (2009).
\newblock Comparing models of strategic thinking in van huyck, battalio, and
  beil's coordination games.
\newblock {\em Journal of the European Economic Association}, 7(2-3):365--376.

\bibitem[Costa-Gomes and Shimoji, 2015]{costa2015comment}
Costa-Gomes, M.~A. and Shimoji, M. (2015).
\newblock A comment on “can relaxation of beliefs rationalize the winner's
  curse?: An experimental study”.
\newblock {\em Econometrica}, 83(1):375--383.

\bibitem[Crawford, 2003]{crawford2003lying}
Crawford, V.~P. (2003).
\newblock Lying for strategic advantage: Rational and boundedly rational
  misrepresentation of intentions.
\newblock {\em American Economic Review}, 93(1):133--149.

\bibitem[Crawford et~al., 2013]{crawford2013}
Crawford, V.~P., Costa-Gomes, M.~A., and Iriberri, N. (2013).
\newblock Structural models of nonequilibrium strategic thinking: Theory,
  evidence, and applications.
\newblock {\em Journal of Economic Literature}, 51(1):5--62.

\bibitem[Crawford and Iriberri, 2007]{crawford2007}
Crawford, V.~P. and Iriberri, N. (2007).
\newblock Level-\textit{k} auctions: Can a nonequilibrium model of strategic
  thinking explain the winner's curse and overbidding in private-value
  auctions?
\newblock {\em Econometrica}, 75(6):1721--1770.

\bibitem[Crawford et~al., 2009]{crawford2009}
Crawford, V.~P., Kugler, T., Neeman, Z., and Pauzner, A. (2009).
\newblock Behaviorally optimal auction design: Examples and observations.
\newblock {\em Journal of the European Economic Association}, 7(2-3):377--387.

\bibitem[Crosetto and Filippin, 2013]{crosetto2013}
Crosetto, P. and Filippin, A. (2013).
\newblock The “bomb” risk elicitation task.
\newblock {\em Journal of Risk and Uncertainty}, 47(1):31--65.

\bibitem[De~Clippel et~al., 2019]{mechanismdesign}
De~Clippel, G., Saran, R., and Serrano, R. (2019).
\newblock Level-\textit{k} mechanism design.
\newblock {\em The Review of Economic Studies}, 86(3):1207--1227.

\bibitem[Filiz-Ozbay and Ozbay, 2007]{ozbay2007}
Filiz-Ozbay, E. and Ozbay, E.~Y. (2007).
\newblock Auctions with anticipated regret: Theory and experiment.
\newblock {\em American Economic Review}, 97(4):1407--1418.

\bibitem[Galavotti et~al., 2018]{observational2}
Galavotti, S., Moretti, L., and Valbonesi, P. (2018).
\newblock Sophisticated bidders in beauty-contest auctions.
\newblock {\em American Economic Journal: Microeconomics}, 10(4):1--26.

\bibitem[Georganas et~al., 2015]{georganas2015}
Georganas, S., Healy, P.~J., and Weber, R.~A. (2015).
\newblock On the persistence of strategic sophistication.
\newblock {\em Journal of Economic Theory}, 159:369--400.

\bibitem[Gillen, 2009]{gillen2009identification}
Gillen, B. (2009).
\newblock Identification and estimation of level-k auctions.
\newblock {\em Available at SSRN 1337843}.

\bibitem[Harsanyi, 1967]{harsanyi1967}
Harsanyi, J.~C. (1967).
\newblock Games with incomplete information played by “bayesian” players,
  i--iii part i. the basic model.
\newblock {\em Management Science}, 14(3):159--182.

\bibitem[Harstad et~al., 1990]{harstad1990}
Harstad, R.~M., Kagel, J.~H., and Levin, D. (1990).
\newblock Equilibrium bid functions for auctions with an uncertain number of
  bidders.
\newblock {\em Economics Letters}, 33(1):35--40.

\bibitem[Holt, 1980]{holt1980competitive}
Holt, C.~A. (1980).
\newblock Competitive bidding for contracts under alternative auction
  procedures.
\newblock {\em Journal of political Economy}, 88(3):433--445.

\bibitem[Horta{\c{c}}su et~al., 2019]{observational1}
Horta{\c{c}}su, A., Luco, F., Puller, S.~L., and Zhu, D. (2019).
\newblock Does strategic ability affect efficiency? evidence from electricity
  markets.
\newblock {\em American Economic Review}, 109(12):4302--42.

\bibitem[Ivanov et~al., 2010]{ivanov2010can}
Ivanov, A., Levin, D., and Niederle, M. (2010).
\newblock Can relaxation of beliefs rationalize the winner's curse?: An
  experimental study.
\newblock {\em Econometrica}, 78(4):1435--1452.

\bibitem[Kagel and Levin, 1995]{kagel1995auctions}
Kagel, J. and Levin, D. (1995).
\newblock Auctions: A survey of experimental research," handbook of
  experimental economics, edited by j. kagel and a. roth.

\bibitem[Kirchkamp and Rei{\ss}, 2011]{kirchkamp2011out}
Kirchkamp, O. and Rei{\ss}, J.~P. (2011).
\newblock Out-of-equilibrium bids in first-price auctions: Wrong expectations
  or wrong bids.
\newblock {\em The Economic Journal}, 121(557):1361--1397.

\bibitem[Klemperer, 1999]{klemperer1999}
Klemperer, P. (1999).
\newblock Auction theory: A guide to the literature.
\newblock {\em Journal of economic surveys}, 13(3):227--286.

\bibitem[Krishna, 2009]{krishna2009}
Krishna, V. (2009).
\newblock {\em Auction theory}.
\newblock Academic press.

\bibitem[Nagel, 1995]{nagel}
Nagel, R. (1995).
\newblock Unraveling in guessing games: An experimental study.
\newblock {\em The American Economic Review}, 85(5):1313--1326.

\bibitem[Noussair and Silver, 2006]{noussair2006behavior}
Noussair, C. and Silver, J. (2006).
\newblock Behavior in all-pay auctions with incomplete information.
\newblock {\em Games and Economic Behavior}, 55(1):189--206.

\bibitem[Rabin, 2000]{rabin}
Rabin, M. (2000).
\newblock Risk aversion and expected-utility theory: A calibration theorem.
\newblock {\em Econometrica}, 68(5):1281--1292.

\bibitem[Rasooly and Gavidia-Calderon, 2020]{rasooly}
Rasooly, I. and Gavidia-Calderon, C. (2020).
\newblock The importance of being discrete: on the (in-) accuracy of continuous
  approximations in auction theory.
\newblock {\em arXiv preprint arXiv:2006.03016}.

\bibitem[Stahl and Wilson, 1994]{stahl1994experimental}
Stahl, D.~O. and Wilson, P.~W. (1994).
\newblock Experimental evidence on players' models of other players.
\newblock {\em Journal of economic behavior \& organization}, 25(3):309--327.

\bibitem[Thaler, 2015]{thaler2015}
Thaler, R.~H. (2015).
\newblock {\em Misbehaving: The Making of Behavioral Economics}.
\newblock W.W. Norton \& Company.

\end{thebibliography}

\newpage

\begin{appendices}

\section{Proofs} \label{proofs}

\begin{proof}[Proof of Proposition 1]
Since our game has a symmetric and finite normal form, it must have a symmetric equilibrium (see \cite{harsanyi1967} and \cite{cheng} for details). Turning to uniqueness, we now generalise three lemmas from \cite{rasooly} to the case of mixed strategies. Let $\sigma$ denote an arbitrary symmetric equilibrium strategy; and for any $v \in \mathbb{X}$, let $s(v) \subseteq \mathbb{X}$ denote the set of bids that are submitted with positive probability at value $v$ (i.e. the support at $v$). 
\begin{lemma}
In any symmetric equilibrium $\sigma$, $s(0) = \{ 0 \}$.
\end{lemma}
\begin{proof}
Consider a player with $v = 0$. If they bid $0$, they get (regardless of their opponents' strategies) a payoff $\pi(v = 0, b = 0) = 0 \times \mathbb{P\text{(win}}|b=0) - 0 = 0$. On the other hand, if they submit any bid $b \geq 1$, they get (given any opponent strategy profile) $\pi(v = 0, b) = 0 \times \mathbb{P\text{(win}}|b) - b = - b < 0$. Thus, bidding $0$ strictly dominates bidding any $b \geq 1$; which means that $\mathbb{P}(b = 0|v = 0) = 1$ in any equilibrium (and in any symmetric equilibrium).\end{proof}

\begin{lemma}[Monotonicity]
Consider any symmetric equilibrium $\sigma$ and consider any two values $v, v' \in \mathbb{X}$ with $v > v'$. Then for any $b \in s(v)$ and $b' \in s(v')$, $b \geq b'$.
\end{lemma}
\begin{proof}
For contradiction, let us suppose that $b < b'$. Since $b$ is submitted with positive probability when the value is $v$, bidding $b$ must be weakly optimal. In particular, it must be weakly better than bidding $b'$:
\begin{equation}\label{warp1}
\begin{split}
&\pi(v, b) \geq \pi(v, b') \\ \iff  & v \mathbb{P\text{(win}}|b) - b \geq v \mathbb{P\text{(win}}|b') - b' \\
\iff & b' - b \geq v \left(\mathbb{P\text{(win}}|b') - \mathbb{P\text{(win}}|b)\right) .
\end{split}
\end{equation}
Similarly, since $b'$ is submitted with positive probability when the value is $v'$, it must be weakly better than $b$:
\begin{equation}\label{warp2}
\begin{split}
&\pi(v', b') \geq \pi(v', b) \\ \iff  & v' \mathbb{P\text{(win}}|b') - b' \geq v' \mathbb{P\text{(win}}|b) - b \\
\iff & v'\left(\mathbb{P\text{(win}}|b') -  \mathbb{P\text{(win}}|b) \right) \geq b' - b.
\end{split}
\end{equation}
Inequalities (\ref{warp1}) and (\ref{warp2}) jointly imply that
\begin{equation}\label{crucial}
v'\left(\mathbb{P\text{(win}}|b') -  \mathbb{P\text{(win}}|b) \right) \geq v\left(\mathbb{P\text{(win}}|b') - \mathbb{P\text{(win}}|b)\right) .
\end{equation}
Since $b' > b$, and $b$ is a symmetric equilibrium bid, it must be that $\mathbb{P\text{(win}}|b') > \mathbb{P\text{(win}}|b)$ (for example, one wins with a bid of $b'$, but not with a bid of $b$, if all of one's opponents bid $b$ -- and this happens with positive probability). Since $\mathbb{P\text{(win}}|b') - \mathbb{P\text{(win}}|b) > 0$, inequality \ref{crucial} then yields $v' \geq v$, which contradicts our initial assumption that $v' < v$. This establishes that $b \geq b'$ as claimed.\end{proof}

\begin{lemma}[No gaps]
In any symmetric equilibrium $\sigma$, the bids that are submitted with positive probability are a set of consecutive integers.
\end{lemma}
\begin{proof}
Suppose for contradiction that this were false. Then there must exist two bids $b_{high} \geq b_{low} + 2$ that are submitted with positive probability even though no bid in between $b_{high}$ and $b_{low}$ is submitted with positive probability. Now if types were to deviate from bidding $b_{high}$ to bidding $b_{low} + 1$, their payment would fall. However, their probability of winning would remain unchanged: both before and after the deviation, they win the auction if and only if all opponents bid $b_{low}$ or lower. So this deviation would be strictly profitable, implying that this could not have been an SE.\end{proof}

We now introduce the concept of jump form.

\begin{definition}
For every bid $i \in B^\sigma$, we define the $i$th jump $j_i$ as $$j_i = v_i + \mathbb{P}(b < i|v_i)$$
where $v_i$ is the minimum value $v \in \mathbb{X}$ such that $\mathbb{P}(b = i|v) > 0$. We refer to the vector of jumps $j = (j_0, j_1, ..., j_m)$ as a jump vector and say that it is increasing if $j_{i+1} > j_i$ for every $i \in \{0, 1, ..., m-1\}$.
\end{definition}

Under the lemmas outlined previously, this is an equivalent representation of behavioural strategies.

\begin{lemma}
There is a bijection between the set of gapless and monotone strategies that satisfy $s(0) = 0$ and the set of increasing jump vectors.
\end{lemma}

\begin{proof}
To establish the first claim, take any gapless and monotone strategy that specifies $\mathbb{P}(b = 0|v = 0) = 1$ and convert it into a jump vector $j$. To show that $j$ is increasing, consider any two consecutive jump points $j_i, j_{i+1} \in j$. If the strategy is monotone, every value that submits a bid $b = i + 1$ must be weakly larger than the every value that submits a bid $b = i$. So in particular, $v_{i + 1} \geq v_i$. This yields two possibilities: either $v_{i+1} > v_i$ or $v_{i+1} = v_i$. In the first instance ($v_{i+1} > v_i$),
\begin{equation}
\begin{split}
    j_{i + 1} &\equiv v_{i + 1} + \mathbb{P}(b < i + 1|v_{i + 1}) \\
    &\geq v_{i} + 1 + \mathbb{P}(b < i + 1|v_{i + 1}) \\
    &\geq  v_{i} + 1 \\
    & > v_i + \mathbb{P}(b < i|v_i) \equiv j_i
\end{split}
\end{equation}
where the final inequality holds since $\mathbb{P}(b < i|v_i) < 1$ (which in turn holds since necessarily $\mathbb{P}(b = i|v_i) > 0$). So in that case, $j_{i+1} > j_{i}$. In the second instance ($v_{i+1} = v_i$), we have
\begin{equation}
\begin{split}
    j_{i + 1} &\equiv v_{i + 1} + \mathbb{P}(b < i + 1|v_{i + 1}) \\
    &= v_{i} + \mathbb{P}(b < i + 1|v_{i}) \\
    &= v_{i} + \mathbb{P}(b < i|v_{i}) + \mathbb{P}(b = i|v_{i}) \\
    & > v_i + \mathbb{P}(b < i|v_i) \equiv j_i
\end{split}
\end{equation}
where the inequality again holds since $\mathbb{P}(b = i|v_{i}) > 0$. Either way, then, we see that monotonicity implies that $j_{1 + 1}> j_i$. Since this true for all $i = 2, 3, ..., m$, this then implies that that vector of jumps is increasing. So every monotone behavioural strategy generates an increasing vector of jumps.

We now argue that any increasing vector of jumps $j = (j_1, ..., j_m)$ is consistent with a unique monotone and gapless behavioural strategy that satisfies $\mathbb{P}(b = 0|v = 0)$. First, use $j$ to compute the minimum values that submit each of the possible bids $v_i$ using the fact that $v_i = \floor{j_i}$ for all $i = 1, ..., m$. (Notice that, since $j_{i+1} > j_i$ for every $i$, $v_{i+1} = \floor{j_{i+1}} \geq v_i = \floor{j_i}$, i.e. the $v_i$ are weakly increasing.) We will begin by showing how this uniquely determines the behavioural strategy for all values other than the $v_i$.

First, consider all $v \in \mathbb{X}$ with $v < v_1$. Recall that $v_1$ is the minimum $v \in \mathbb{X}$ such that $\mathbb{P}(b  = 1|v) > 0$. So $\mathbb{P}(b = 1|v_1) > 0$; and $\mathbb{P}(b = 1|v) = 0$ for any value $v < v_1$. Since $\mathbb{P}(b = 1|v_1) > 0$, monotonicity implies that $\mathbb{P}(b > 1|v) = 0$ for any $v < v_1$. Since $\mathbb{P}(b = 1|v) = 0$ and $\mathbb{P}(b > 1|v) = 0$ for all $v < v_1$, $\mathbb{P}(b = 0|v) = 1$ for all $v < v_1$, i.e. all such values bid $0$ with probability $1$.

Next, consider now all $v \in \mathbb{X}$ with $v > v_m$ (where $m$ is the largest bid submitted). By definition, $v_m$ is the minimum value $v$ such that $\mathbb{P}(b = m|v) > 0$. So $\mathbb{P}(b = m|v_m) > 0$. By monotonicity, this means that $\mathbb{P}(b \geq m|v) = 1$ for all $v > v_m$. But $m$ is the largest bid submitted with positive probability. Hence, $\mathbb{P}(b = m|v) = 1$ for all $v > v_m$, i.e. all such values bid $m$ with probability $1$.

Finally, consider all remaining values $v$ such that $v \neq v_i$ for all $i = 1, ..., m$. For any such value, we can find a bid $i$ such that $v_i < v < v_{i + 1}$. Using the previous two arguments, one can show that $\mathbb{P}(b \geq i|v) = 1$ and $\mathbb{P}(b \leq i|v) = 1$. But then $\mathbb{P}(b = i|v) = 1$, i.e. all such values bid $i$ with probability $1$.

It remains to consider the values $v \in \mathbb{X}$ such that $v = v_i$ for some $i = 1, 2, ..., m$. Consider then any such value and let $B^v$ denote the set of bids $i$ such that $v_i = v$. Since the $v_i$ are weakly increasing (in the bid $i$), $B^v$ must be a set of consecutive integers (or a set containing a single integer). Label these integers $b', b' + 1, ..., b' + k$. For all such bids, $v$ is the minimum value such that $\mathbb{P}(b|v) > 0$. So all such bids must belong to the support of the strategy at value $v$.

We now argue that $\mathbb{P}(b \geq b' - 1|v) = 1$, i.e. bids strictly lower than $b' - 1$ cannot belong to the support. To see this, note that, since the $v_i$ are increasing, $v_{b'} \geq v_{b' - 1}$. However, $v_{b'} = v \neq v_{b' - 1}$: otherwise, the bid $b' - 1$ would belong to $B^v$. Hence,  $v_{b'} = v > v_{b' - 1}$. So the minimum value that bids $b' - 1$ is strictly smaller than $v$. By monotonicity, this means that the value $v$ must bid at least $b' - 1$ with probability $1$.

Next, we argue that $\mathbb{P}(b >  b' + k|v) = 0$. The argument is similar. Since the $v_i$ are increasing, $v_{b''} \geq v_{b' + k}$ for any higher bid $b'' >  b' + k$. However, $v_{b''} \neq v_{b' + k} = v$: otherwise, the bid $b'$ would be in $B^v$. So $v_{b''} > v_{b' + k}$. So the smallest value that bids $b''$ is strictly larger than $v$. As a result, $\mathbb{P}(b = b''|v) = 0$.

Given the previous arguments, we know that the support at $v$ must contain $B^v = \{b', b' + 1, ..., b' + k\}$. In addition, the only other bid that may belong to the support is $b' - 1$. So, for every bid $i \in B^v$, we can write
\begin{equation}
\begin{split}
j_i &= v_i + \mathbb{P}(b < i|v) = v + \sum_{j = b' - 1}^{i-1}\mathbb{P}(b = j|v) \\
\end{split}
\end{equation}
Thus, we have the system of equations
\begin{equation}
\begin{split}
j_{b'} - v &= \mathbb{P}(b = b' -  1|v) \\
j_{b' + 1} - v &= \mathbb{P}(b = b' -  1|v) + \mathbb{P}(b'|v) \\
. \\
. \\
j_{b' + k} - v &= \mathbb{P}(b = b' -  1|v) + \mathbb{P}(b = b'|v) + ... + \mathbb{P}(b = b' + k - 1)
\end{split}
\end{equation}
Given the vector $j$, the first equation uniquely determines $\mathbb{P}(b = b' -  1|v)$ (which may equal zero). Using the value obtained for $\mathbb{P}(b = b' -  1|v)$, the second equation uniquely determines $\mathbb{P}(b = b'' -  1|v)$. Continuing in this manner, the system of equations uniquely pins down $\mathbb{P}(b = i|v)$ for $i = b' - 1, b', ..., b' + k - 1$. Given the previous arguments, the only other bid in the support is $b = b' + k$ (the maximum bid). This is then also uniquely determined using
\begin{equation}
\mathbb{P}(b = b' + k|v) = 1 - \sum_{i = b'' - 1}^{b' - 1}\mathbb{P}(b = i|v)
\end{equation}
The argument above shows how a vector of jumps uniquely pins down a probability distribution over bids for any value $v \in \mathbb{X}$. Repeating the argument for every such value, we see that it is consistent with a unique behavioural strategy.\end{proof}

Having recast strategies in jump form, we can now define our algorithm:

\begin{definition}
Let $\hat{j} = (\hat{j_1}, ... \hat{j}_m)$ denote the output of the following algorithm:
\begin{enumerate}
    \item Impose $j_0$ = 0.
    \item Starting at $i = 1$, find the minimum $j_i \in (j_{i-1}, S]$ such that $\pi(v = \floor{j_i}, b = i) - \pi(v = \floor{j_i}, b = i - 1) \geq 0$.
    \item If there is no such $j_i$, the algorithm terminates.
    \item There there does exist such a $j_i$, repeat (from step 2) for $i + 1$.
\end{enumerate}
\end{definition}

We then observe the following.
\begin{lemma}
The vector $\hat{j}$ is the \textit{only} possible symmetric equilibrium of the all-pay auction.
\end{lemma}
\begin{proof}
Suppose for contradiction that there were some symmetric equilibrium $j \neq \hat{j}$. If $j \neq \hat{j}$, there must be some jump at which they differ. Let $j_i$ denote the first such jump (i.e. $i$ is the smallest number such that $j_i \neq \hat{j_i})$. There are two possibilities: $j_i < \hat{j_i}$ or instead $j_i > \hat{j_i}$. We will argue that each leads to a contradiction.

To see why the first case is impossible, first define $\pi^{\hat{j}}(v, b)$ and $\pi^j(v, b)$ as the payoff of a player with a value $v$ who bids $b$ given that their opponents all bid according to $\hat{j}$ and $j$ respectively. Now recall that, by construction, $\hat{j_i}$ is the \textit{minimum} jump such that
\begin{equation}
\pi^{\hat{j}}(v = \floor{\hat{j_i}}, b = i) \geq \pi^{\hat{j}}(v = \floor{\hat{j_i}}, b = i-1)
\end{equation}
Since $j_i < \hat{j_i}$, this means that
\begin{equation}
\pi^{j}(v = \floor{j_i}, b = i) < \pi^{j}(v = \floor{j_i}, b = i-1)
\end{equation}
But then, given that one's opponents bid according to $j$, one should never bid $b = i$ with a value $v = \floor{j_i}$. That is, $\mathbb{P}(b = i|v = \floor{j_i}) = 0$. But then $v = \floor{j_i}$ cannot be the smallest value that bids $i$ (contradiction).

To see why the second case is impossible, recall again that every $\hat{j_i}$ satisfies
\begin{equation}
\begin{split}
\pi^{\hat{j}}(v = \floor{\hat{j_i}}, b = i) &\geq \pi^{\hat{j}}(v = \floor{\hat{j_i}}, b = i-1) \iff \\
\floor{\hat{j_i}}\mathbb{P}^{\hat{j}}(\text{win}|b = i) - 1 &\geq \floor{\hat{j}_{i}}\mathbb{P}^{\hat{j}}(\text{win}|b = i-1)
\end{split}
\end{equation}
As a preliminary, note that since the RHS is non-negative, $\floor{\hat{j_i}}\mathbb{P}^{\hat{j}}(\text{win}|b = i) - 1 \geq 0$ and so $\floor{\hat{j_i}} \geq 1$ for every $i$. Let us record this fact.

Turning to the main argument, fix the value at $\hat{j_i}$ but consider now changing opponent bidding from $\hat{j}$ to $j$. If $j_i > \hat{j_i}$, then $\mathbb{P}^j(\text{win}|b = i) > \mathbb{P}^{\hat{j}}(\text{win}|b = i)$. However, since $j$ and $\hat{j}$ agree for all previous jumps, $\mathbb{P}^j(\text{win}|b = i-1) = \mathbb{P}^{\hat{j}}(\text{win}|b = i-1)$. Since $\hat{\floor{j_i}} \geq 1$ (established earlier), this all means that
\begin{equation}
\begin{split}
\floor{\hat{j_i}}\mathbb{P}^{j}(\text{win}|b = i) - 1 &> \floor{\hat{j}_{i}}\mathbb{P}^{j}(\text{win}|b = i-1) \iff \\
\pi^{j}(v = \floor{\hat{j_i}}, b = i) &> \pi^{j}(v = \floor{\hat{j_i}}, b = i-1)
\end{split}
\end{equation}
If $j$ is an SE, this implies that $\mathbb{P}^j(b = i-1|v = \floor{\hat{j_i}}) = 0$. However, this is impossible: since $j_i > \hat{j}_i$, $\mathbb{P}^j(b = i-1|v = \floor{\hat{j_i}}) = \min\{j_i, \floor{\hat{j}_i} + 1\} - \max \{\hat{j}_{i-1}, \floor{\hat{j}_i}\} > \hat{j}_i - \max \{\hat{j}_{i-1}, \floor{\hat{j}_i}\} = \mathbb{P}^{\hat{j}}(b = i-1|v = \floor{\hat{j_i}}) \geq 0$ and so $\mathbb{P}^j(b = i-1|v = \floor{\hat{j_i}}) > 0$.\end{proof}

By the argument above, there is no possible symmetric equilibrium except for $\hat{j}$. At we noted earlier, however, our game must possess some symmetric equilibrium. From this, it follows that it has exactly one symmetric equilibrium (namely, $\hat{j}$). This concludes the proof.\end{proof}

\begin{proof}[Proof of Proposition 2]
See main text.\end{proof}

\begin{proof}[Proof of Proposition 3]
A proof of this result can be easily reconstructed from the proof of Proposition \ref{prop1} (and is available from the author upon request).\end{proof}

\begin{proof}[Proof of Proposition 4]
We argue by induction. First, we show that the desired pattern of bidding holds when $k = 1$. Next, we show that if it holds for any level $k - 1  \in \mathbb{K}$, then it also holds (provided that $x$ is sufficiently large) for any level $k \in \mathbb{K}$.

Starting with the base case ($k = 1$), note that level-$1$ bidders choose a bid $b \in \mathbb{X}$ to maximise \vspace{-1em}
\begin{equation}
\pi^1 (v, b) = (v - b)\left(\frac{n-1}{n} + \frac{b}{n(x+1)}\right)^{n-1}
\end{equation}
(inserting $p = 1/n$ into \ref{formula}). Considering this function's continuous extension, observe that
\begin{equation}\label{log_derivative}
\frac{\partial \ln [\pi^1 (v, b)]}{\partial b} = \frac{(n-1)(v - (x+1)) - nb}{(v-b)(b + (n-1)(x+1))}
\end{equation}
Since are searching for the optimal bid, we can just consider the range $b \leq v$. In this range, the denominator of (\ref{log_derivative}) is non-negative. Moreover, since $v \leq x$ for all $v \in \mathbb{X}$, the numerator is negative. Thus, $\partial  \ln[\pi^1 (v, b)] / \partial b \leq 0$ and so $\partial  \pi^1 (v, b) / \partial b \leq 0$. Since this is true of the function's continuous extension, it follows that $\pi(v, 0) \geq \pi(v, b)$ for every value $v \in \mathbb{X}$ and every positive integer bid $b \in \mathbb{X}$. That is, bidding zero is optimal (as claimed).

Next, we show that, if the desired pattern of bidding holds for level $k - 1  \in \mathbb{K}$, then it also holds (provided that $x$ is sufficiently large) for level $k \in \mathbb{K}$. Suppose then that all players of level $k-1$ bid as proposed. Then they bid $0$ when $v \leq v^*(k-1)$, i.e. at $\floor{v^*(k-1) + 1} $ distinct values; and otherwise, they bid $k - 2$. Hence, for a level-$k$ player,
\begin{equation}
\mathbb{P}^k(\text{win}|b) = 
\begin{cases}
1 & \text{if} \hspace{0.5em} b > k - 2 \\
\left(1-p + p\left(\frac{\floor{v^*(k-1) + 1}}{x+1}\right) \right)^{n-1}& \text{if} \hspace{0.5em} b \in \{1, ..., k - 2\} \\
(1-p)^{n-1} & \text{if} \hspace{0.5em} b = 0 \\
\end{cases}
\end{equation}
Given this, it is immediate that $b^* \in \{0, 1, k - 1\}$, i.e. there are at most three optimal bids. Moreover, $\pi^k(v, b = 0) - \pi^k(v, b = 1) \rightarrow v(1-p)^{n-1} - (v-1)(1-p)^{n-1} < 0$ as $x \rightarrow \infty$, and so $b^* \in \{0, k - 1\}$ if $x$ is sufficiently large. Finally, observe that the latter is optimal iff
\begin{equation}
\begin{split}
\pi(v, b = k-1) > \pi(v, b = 0) &\iff v - (k-1) > v(1-p)^{n-1} \\
&\iff v > \frac{k-1}{1 - (1-p)^{n-1}}
\end{split}
\end{equation}
i.e. iff $v > v^*(k)$ (as claimed).\end{proof}
\newpage

\section{Illustrating the algorithm} \label{algorithm}

In this section, we illustrate the algorithm using two examples.  In both, values are uniform so (in any symmetric strategy profile)
\begin{equation}
\begin{split}
\mathbb{P}(b < i) &= \mathbb{P}(v < \floor{j_i}) + \mathbb{P}(v = \floor{j_i})(j_i - \floor{j_i}) \\
&= \frac{\floor{j_i}}{S} + \frac{1}{S}(j_i - \floor{j_i}) \\
&= \frac{j_i}{S}
\end{split}
\end{equation}
where $S$ is the number of possible valuations.

\textbf{Example 1.} Consider an all-pay auction with uniform values and $n = 2$. To find the first jump point, we look for the minimum $j_1 \in [2, S]$ such that
\begin{equation}
\pi(v = \floor{j_1}, b = 1) - \pi(v = \floor{j_1}, b = 0) \geq 0 
\end{equation}
Clearly, $\pi(v = \floor{j_1}, b = 0) = 0$ since $\mathbb{P\text{(win}}|b=0) = 0$. In addition, $\mathbb{P\text{(win}}|b=1) = \mathbb{P}(b = 0) = j_1/S$ and so $\pi(v = \floor{j_1}, b = 1) = \floor{j_1}j_1/S - 1$. So we look for the minimum $j_1$ such that
\begin{equation}
    \frac{\floor{j_1}j_1}{S} - 1 \geq 0
\end{equation}
Of course, the solution will depend on $S$. If $S = 100$, $j_1 = 10$ (and similarly $j_1 = \sqrt S$ for any square $S$). If $S = 99$, the solution remains at $j_1 = 10$. However, if $S = 101$, the solution becomes $j_1 = 10.1$, implying randomisation. (One can prove that, as $S \rightarrow \infty$, the proportion of solutions that are integer solutions converges to $1/2$.) Let us suppose that $S = 101$ so that the first jump is $j_1 = 10.1$.

To find the next jump, we look for the smallest $j_2$ that satisfies
\begin{equation}
\begin{split}
&\pi(v = \floor{j_2}, b = 2) - \pi(v = \floor{j_2}, b = 1) \geq 0 \\
\iff & \floor{j_2} \left(\mathbb{P\text{(win}}|b=2) - \mathbb{P\text{(win}}|b=1)\right) \geq 1 \\
\iff & \floor{j_2} \left(\mathbb{P}(b \leq 1) - \mathbb{P}(b = 0)\right) \geq 1 \\
\iff & \floor{j_2} \left(\frac{j_2}{101} - \frac{10.1}{101}\right) \geq 1 \\
\end{split}
\end{equation}
One can check that the solution is $j_2 = 16.4125$, so the equilibrium vector of jumps starts with $(0, 10.1, 16.4125, ...)$. Returning to behavioural strategies, this means that all values $v \in \{0, 1, ..., 9 \}$ bid $0$ (with probability $1$), value $v = 10$ randomises between $b = 0$ and $b = 1$ with probabilities $0.1$ and $0.9$, values $v \in \{11, ..., 15\}$ bid $1$ (with probability $1$), and finally the value $v = 16$ bids $b = 1$ with probability $0.4125$. (Since we haven't computed $j_3$, we cannot technically determine the probability with which value $v = 16$ bids $v = 2$; but in fact $j_3 \geq 17$ so this probability is $1 - 0.4125$.)

\textbf{Example 2.} Consider now the first price auction with cancelled bids. To find the first jump $j_1$, we look for the smallest $j_1 \in (0, S]$ such that
\begin{equation}\label{example}
\begin{split}
&\pi(v = \floor{j_1}, b = 1) \geq \pi(v = \floor{j_1}, b = 0) \\
\iff & (\floor{j_1} - 1)\mathbb{P\text{(win}}|b=1) \geq \floor{j_1}\mathbb{P\text{(win}}|b=0) \\
\iff & (\floor{j_1} - 1)(1-p + p\frac{j_1}{S})^{n-1}p \geq \floor{j_1}(1-p)^{n-1}p \\
\iff & (\floor{j_1} - 1)(1-p + p\frac{j_1}{S})^{n-1} \geq \floor{j_1}(1-p)^{n-1} \\
\end{split}
\end{equation}
For instance, suppose that $p = 1/2$, $n = 2$ and $S = 101$. Then (\ref{example}) reduces to 
\begin{equation}
\floor{j_1}j_1 - 101 - j_1 \geq 0
\end{equation}
which is almost quadratic in $j_1$. One can check that this inequality is satisfied by $j_1 = 11$: then $\floor{j_1}j_1 - 101 - j_1 = 11^2 - 101 - 11 = 9 > 0$. On the other hand, it won't hold if $j_1 < 11$: for then $\floor{j_1}j_1 - 101 - j_1 \leq 10 \times 11 - 101 - 11 = -2 < 0$. So the minimum $j_1$ that satisfies the inequality is $11$; this is the first jump. This means that all values smaller than $11$ must bid zero (with probability $1$); but the value $v = 11$ never bids zero with any probability.

\newpage

\section{Optimising $p$} \label{optimising}

To find the probability $p^* \in [0, 1]$ that maximises $d(p)$, recall that $p^* > 0$ (from Proposition \ref{prop4}). Moreover, it is obvious that $p^* < 1$: if $p = 1$, then $d(p) = 0$ (whereas $d(p) > 0$ is plainly possible). Hence, $p^* \in (0, 1)$, which means that the optimal $p$ must satisfy the first-order condition
\begin{equation}\label{foc1}
    \frac{\partial}{\partial p} \int_0^x \left| \beta(v) - \beta^{1}(v) \right| dv = 0.
\end{equation}
One can check that $\beta(v) \geq \beta^{1}(v)$ for all $v \in [0, x]$. So (\ref{foc1}) is equivalent to
\begin{equation}\label{foc2}
    \frac{\partial}{\partial p} \int_0^x  \beta(v) dv = \frac{\partial}{\partial p} \int_0^x  \beta^1(v) dv.
\end{equation}
When $n = 2$, standard calculations reveal that
\begin{align}\label{two_bidders}
\frac{\partial}{\partial p} \int_0^x  \beta(v) dv &= \frac{\partial}{\partial p} \int_0^x \frac{p v^2}{2 (p v - p + 1)}  dv \nonumber \\
&= \frac{\partial}{\partial p} \frac{x^2 \left(p (3 p-2) - 2 (p-1)^2 \ln (1-p)\right)}{4 p^2} \nonumber \\
&= \frac{x^2 ((2-p) p - 2 (p-1) \ln (1-p))}{2 p^3}
\end{align}
Similarly, when $n \geq 3$, we find that
\begin{align}\label{three_bidders}
\frac{\partial}{\partial p} \int_0^x  \beta(v) dv
&= \frac{\partial}{\partial p} \int_0^x \left(\frac{n-1}{n}\right)v - \frac{x(1 - p)}{np}\left[ 1 - \left(\frac{1 - p}{1 - p + p(v/x)}\right)^{n-1} \right]  dv \nonumber \\
&= \frac{\partial}{\partial p}  \frac{x^2 \left(2+n p ((n-1) p-2)\right)-2 (1-p)^n}{2 (n-2) n p^2} \nonumber \\
&= \frac{x^2 \left(p \left(n (1-p)^n-2 (1-p)^n+n+2\right)+2
   \left((1-p)^n-1\right) -n p^2 \right)}{(n-2) n (p-1) p^3}
\end{align}
Turning to the level-$1$ bidding function, we find that
\begin{align}\label{level1}
\frac{\partial}{\partial p} \int_0^x  \beta^1(v) dv &= \frac{\partial}{\partial p} \int_{v^*}^x  \left( \frac{n - 1}{n} \right)v - \left(\frac{1 - p}{p}\right)\frac{x}{n} \hspace{0.3em} dv  \nonumber \\
&= \frac{\partial}{\partial p} \frac{x^2 (n p-1)^2}{2 (n-1) n p^2} \nonumber \\
&= \frac{x^2 (n p-1)}{(n-1) n p^3}
\end{align}
where $v^* = (1-p)x/(n-1)p$.

To find the optimal probability when $n = 2$, we now substitute (\ref{two_bidders}) and (\ref{level1}) into the first-order condition (\ref{foc2}), obtaining
\begin{equation}
\frac{(2-p) p - 2 (p-1) \ln (1-p)}{2 p^3} = \frac{n p-1}{(n-1) n p^3}
\end{equation}
Similarly, when $n \geq 3$, we substitute (\ref{three_bidders}) and (\ref{level1}) and into the first-order condition, yielding
\begin{equation}
\frac{p \left(n (1-p)^n-2 (1-p)^n+n+2\right)+2
   \left((1-p)^n-1\right) -n p^2 }{(n-2) n (p-1) p^3} =  \frac{n p-1}{(n-1) n p^3}
\end{equation}
Either way, we obtain an equation that can be solved numerically to reveal candidates for $p^*$. It is then a straightforward matter to check which is of these candidates generate a larger distance. Since our problem must have a global maximiser (by the Weierstrass theorem), and the global maximiser must satisfy the first-order condition, the $p$ satisfying the first-order condition that generates the largest distance must be the global maximiser.

\begin{table}[h!]
\centering
\caption{Optimal cancellation probabilities (rounded)}
\label{probabilities}
\begin{tabular}{llllllll}
\hline
$n$   & 2     & 3     & 4     & 5     & 6     & 7     & 8     \\ \hline
$p^*$ & 0.536 & 0.343 & 0.256 & 0.204 & 0.170 & 0.145 & 0.127 \\ \hline
\end{tabular}
\end{table}

Table \ref{probabilities} outlines for the solutions for $n \in \{1, ..., 8\}$. As can be seen, the solution is generally quite close to $1/n$ — so the bound identified in Proposition \ref{prop4} comes close to identifying the exact solution. Note also that, as might be expected, the optimal probability $p$ does not depend on the scale parameter $x$.

\newpage

\section{Tables and figures} \label{tables}

\begin{figure}[H]
    \centering
    \caption{The first-price auction when bids are multiples of 5}
    \vspace{-0.2em}
    \includegraphics[width=14cm]{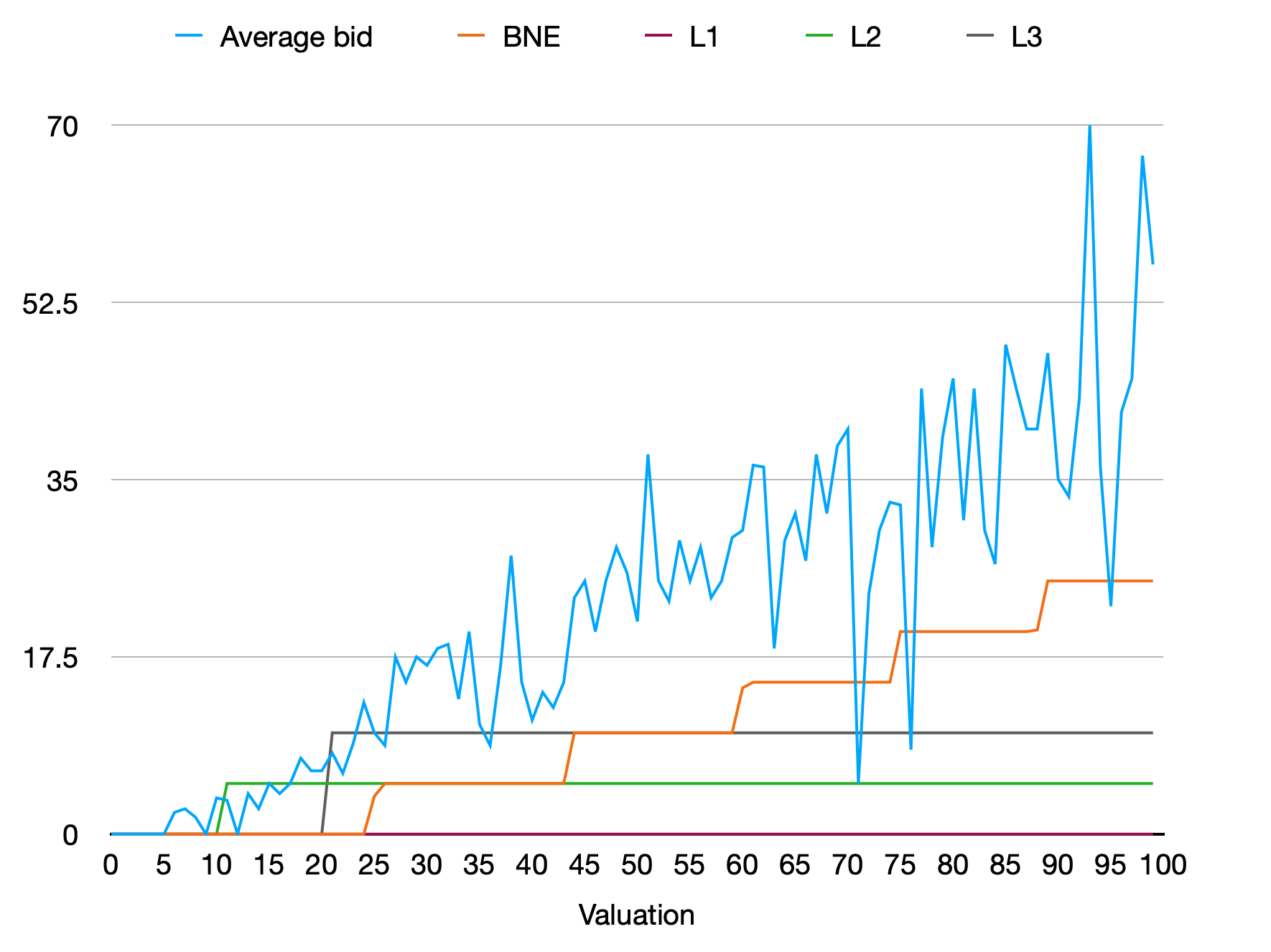}
    \label{t2fp}
\end{figure}

\begin{figure}[H]
    \centering
    \caption{The all-pay auction when bids are multiples of 5}
    \vspace{-0.2em}
    \includegraphics[width=14cm]{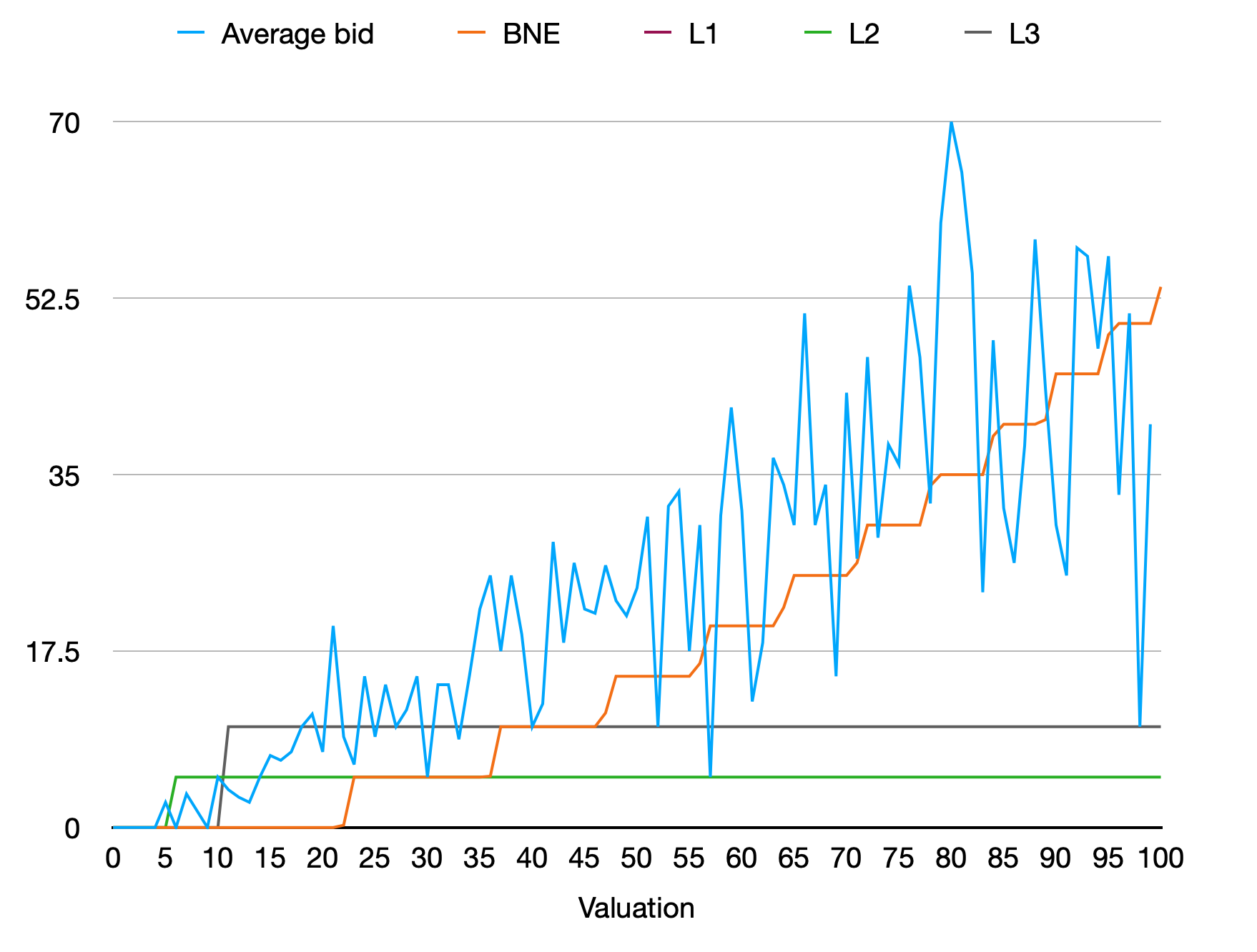}
    \label{t2ap}
\end{figure}

\begin{center}
\begin{table}[H]
\caption{Balance Table}
\label{balance}
\begin{threeparttable}
\begin{tabular}{lccc}
\hline
                 & \multicolumn{1}{l}{Integer bids} & \multicolumn{1}{l}{Multiples of five} & $p$-value \\ \hline
Mean age         & 31.4                             & 32.0                                  & 0.29      \\
Share male       & 0.54                             & 0.43                                  & 0.33      \\
Share maths      & 0.39                             & 0.38                                  & 0.95      \\
Share humanities & 0.24                             & 0.14                                  & 0.36      \\
Share social-sciences     & 0.30                             & 0.29                                  & 0.93      \\ \hline
\end{tabular}
\begin{tablenotes}
\footnotesize
\item \hspace{-0.2em}\textit{Notes}: This table displays how demographics (age, sex and subject) vary between the two treatment groups. The column ‘$p$-value’ reports $t$-tests of the hypothesis that the relevant variable is on average equal across the groups.
\end{tablenotes}
\end{threeparttable}
\end{table}
\end{center}

\begin{table}[H]
\begin{threeparttable}
\caption{Average bids}
\label{average_bids}
\begin{tabular}{lcccccc}
\hline
              & \multicolumn{3}{c}{First-price} & \multicolumn{3}{c}{All-pay}        \\ \hline
              & T1     & T2     & $p$-value   & T1  & T2  & $p$-value            \\ \cline{2-7} 
Average bid   & 23.6    & 23.0    & 0.85        & 22.2 & 25.4 & 0.12                 \\
Average value & 50.9    & 49.6    & 0.27        & 49.4 & 50.8 & 0.89                 \\
Ratio         & 0.46    & 0.46    &             & 0.45 & 0.50 & \multicolumn{1}{l}{} \\ \hline
\end{tabular}
\begin{tablenotes}
\footnotesize
\item \hspace{-0.2em}\textit{Notes}: This table reports the average bids across two treatments and auction structures. The column ‘$p$-value’ reports $t$-tests of the hypothesis that the relevant variable is on average equal across the groups.
\end{tablenotes}
\end{threeparttable}
\end{table}

\begin{table}[H]
\begin{threeparttable}
\caption{Comparing equilibrium and level-$k$}
\label{horse_race}
\begin{tabular}{llcccc}
\hline
      & \textit{} & T1FP    & T1AP    & T2FP    & T2AP    \\ \hline
BNE   & LL        & -4407.9 & -4332.7 & -1024.7 & -1042.2 \\
      & BIC       & 8820.5  & 8670.2  & 2053.1  & 2088.1  \\
L1    & LL        & -4530.4 & -4501.2 & -1088.8 & -1128.6 \\
      & BIC       & 9065.4  & 9007.1  & 2181.3  & 2260.9  \\
L1-L2 & LL        & -4530.4 & -4501.2 & -1085   & -1126.9 \\
      & BIC       & 9070.1  & 9011.7  & 2177.4  & 2261.2  \\
L1-L3 & LL        & -4530.1 & -4501.2 & -1071.1 & -1120.7 \\
      & BIC       & 9074.3  & 9016.4  & 2153.4  & 2252.6  \\ \hline
\end{tabular}
\begin{tablenotes}
\footnotesize
\item \hspace{-0.2em}\textit{Notes}: This table reports the log-likelihoods and associated Bayes-Information Criteria values for four different structural models. In the first model, equilibrium is the only type. The subsequent models are populated by level-1 types, level-1 and 2 types, and level-1, 2 and 3 types respectively.
\end{tablenotes}
\end{threeparttable}
\end{table}

\begin{table}[H]
\begin{threeparttable}
\caption{The hybrid model}
\label{hybrid}
\begin{tabular}{lllllllll}
\hline
           & \multicolumn{2}{c}{T1FP} & \multicolumn{2}{c}{T1AP} & \multicolumn{2}{c}{T2FP} & \multicolumn{2}{c}{T2AP} \\ \hline
\textit{}  & Subject     & Round      & Subject     & Round      & Subject     & Round      & Subject     & Round      \\
$p_0$      & 0.758       & 0.760      & 0.815       & 0.757      & 0.905       & 0.858      & 0.741       & 0.788      \\
           & (0.059)     & (0.046)    & (0.053)     & (0.06)     & (0.066)     & (0.054)    & (0.149)     & (0.083)    \\
$p_1$      & 0.036       & 0.046      & 0.000       & 0.010      & 0.095       & 0.044      & 0.048       & 0.071      \\
           & (0.025)     & (0.02)     & (0.0)       & (0.01)     & (0.066)     & (0.034)    & (0.048)     & (0.04)     \\
$p_2$      & 0.133       & 0.027      & 0.019       & 0.055      & 0.000       & 0.000      & 0.048       & 0.024      \\
           & (0.051)     & (0.016)    & (0.019)     & (0.022)    & (0.0)       & (0.0)      & (0.048)     & (0.024)    \\
$p_3$      & 0.073       & 0.167      & 0.167       & 0.178      & 0.000       & 0.098      & 0.164       & 0.117      \\
           & (0.041)     & (0.041)    & (0.051)     & (0.057)    & (0.0)       & (0.048)    & (0.141)     & (0.073)    \\
$\sigma_0$ & 27.3        & 27.3       & 19.9        & 20.2       & 23.2        & 23.7       & 21.7        & 21.3       \\
           & (1.9)       & (1.8)      & (1.6)       & (1.5)      & (3.0)       & (2.4)      & (2.8)       & (2.0)      \\
$\sigma_1$ & 3.8         & 1.8        & -           & 2.3        & 5.4         & 6.3        & 2.7         & 2.4        \\
           & (2.0)       & (0.3)      & -           & (0.4)      & (2.2)       & (0.7)      & (0.0)       & (0.4)      \\
$\sigma_2$ & 9.2         & 4.2        & 1.1         & 1.0        & -           & -          & 2.7         & 1.0        \\
           & (0.6)       & (0.5)      & (0.0)       & (0.0)      & -           & -          & (0.0)       & (0.0)      \\
$\sigma_3$ & 13.3        & 10.5       & 5.3         & 7.9        & -           & 4.1        & 14.5        & 6.5        \\
           & (1.1)       & (1.2)      & (0.6)       & (2.1)      & -           & (0.9)      & (7.0)       & (4.1)      \\ \hline
LL         & -4241.5     & -4217.1    & -4043.8     & -4004.1    & -987.9      & -982.6     & -962.9      & -940.9     \\ \hline
\end{tabular}
\begin{tablenotes}
\footnotesize
\item \hspace{-0.2em}\textit{Notes}: The `subject’ columns report analyses in which subjects’ levels are fixed; the `round’ columns report analyses in which levels are allowed to vary between rounds. The symbols $p_i$ and $\sigma_i$ respectively denote the probability and noise parameter of type $i \in \mathbb{K}$ (equilibrium is indexed by $i = 0$). LL abbreviates `log-likelihood’ and jack-knife standard errors are in parentheses. When $p_i = 0$, the associated parameter $\sigma_i$ is unidentified and therefore omitted from the table.
\end{tablenotes}
\end{threeparttable}
\end{table}

\begin{figure}[H]
    \centering
    \caption{Comparing estimated levels}
    \vspace{-0.2em}
    \includegraphics[width=14cm]{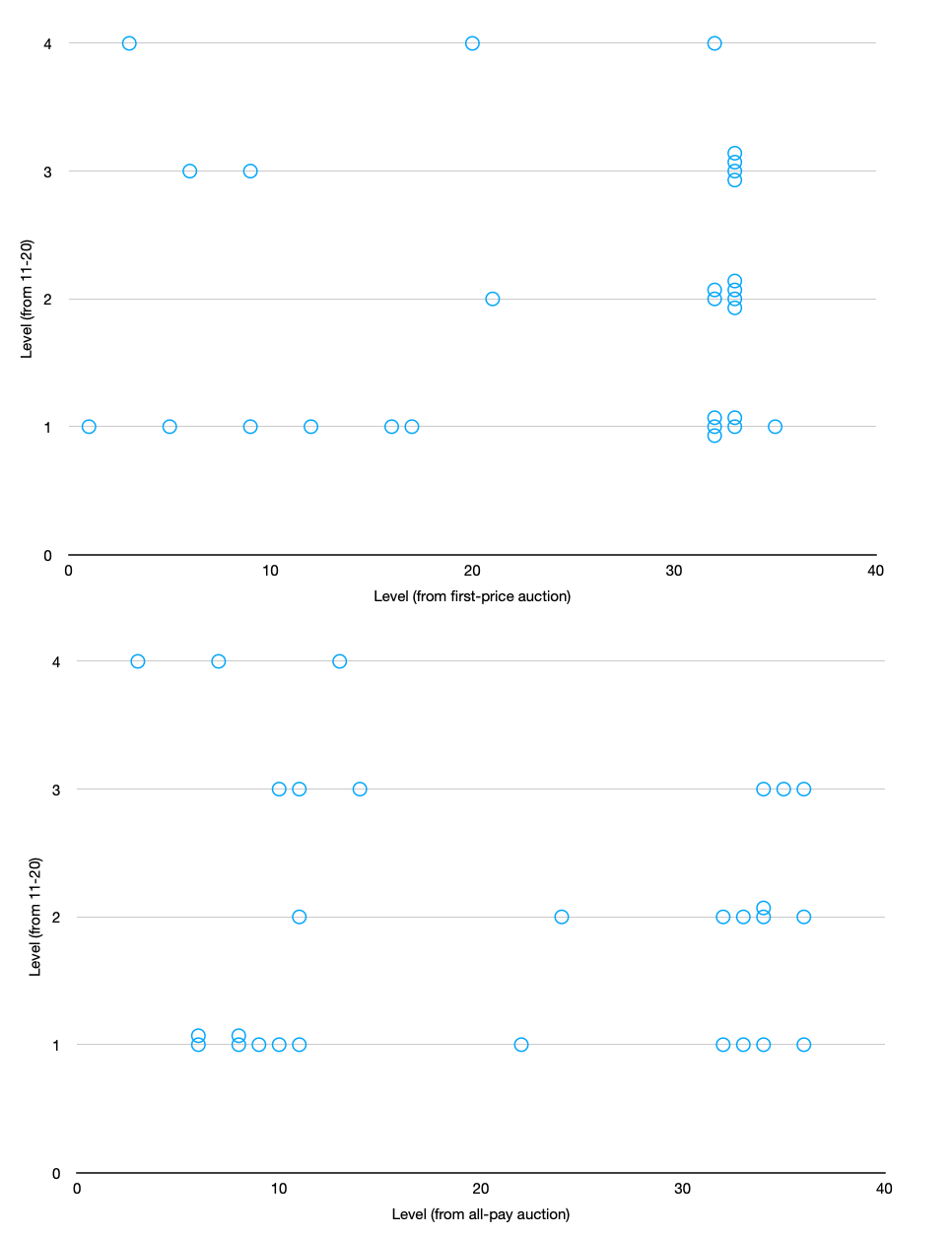}
    \label{correlate}
    \vspace{0.3em}
    \begin{minipage}{12.4cm}%
    \footnotesize \textit{Notes}: The first panel plots levels as estimated from the first-price auction against levels as inferred from the 11-20 game. The estimated correlation between these variables is 0 (95\% confidence interval: -0.37 to 0.37). The second panel plots levels as estimated from the all-pay auction against levels as inferred from the 11-20 game. The estimated correlation between these variables is -0.09 (95\% confidence interval: -0.44 to 0.30).
    \end{minipage}%
\end{figure}

\begin{table}[H]
\begin{threeparttable}
\caption{Robustness checks}
\label{robustness}
\begin{tabular}{clcccc}
\hline
\multicolumn{1}{l}{}                          &           & T1FP & T1AP & T2FP & T2AP \\ \hline
\multirow{2}{*}{\textit{Risk aversion}}       & BNE       & 19.4 & -    & 17.9 & -    \\
                                              & Level-$k$ & 26.9 & -    & 22.7 & -    \\ \hline
\multirow{2}{*}{\textit{Cognitive hierarchy}} & BNE       & 21.7 & 18.2 & 19.7 & 18.8 \\
                                              & Level-$k$ & 31.3 & 23.8 & 26.6 & 29.6 \\ \hline
\multirow{2}{*}{\textit{Dominated bids}}      & BNE       & 22.6 & 19.4 & 19.5 & 19.0 \\
                                              & Level-$k$ & 30.8 & 29.9 & 22.7 & 26.9 \\ \hline
\multirow{2}{*}{\textit{Tie-breaking}}        & BNE       & 21.7 & 18.2 & 19.7 & 18.8 \\
                                              & Level-$k$ & 30.3 & 29.5 & 23.4 & 26.6 \\ \hline
\multirow{2}{*}{\textit{Dropping round 2}}    & BNE       & 18.4 & 18.2 & 18.3 & 17.4 \\
                                              & Level-$k$ & 26.5 & 29.0 & 22.2 & 24.0 \\ \hline
\end{tabular}
\begin{tablenotes}
\footnotesize
\item \hspace{-0.2em}\textit{Notes}: This table reports the root-mean-square prediction errors of the relevant theories following a series of robustness checks. In every case, level-$k$ predictions are obtained by assigning each datapoint the level from the $1-3$ range that minimises the model's prediction error (a procedure that will tend to overstate the model's predictive performance).
\end{tablenotes}
\end{threeparttable}
\end{table}

\begin{figure}[h!]
    \centering
    \caption{Risk aversion (multiples of 5 treatment)}
    \vspace{-0.2em}
    \includegraphics[width=14cm]{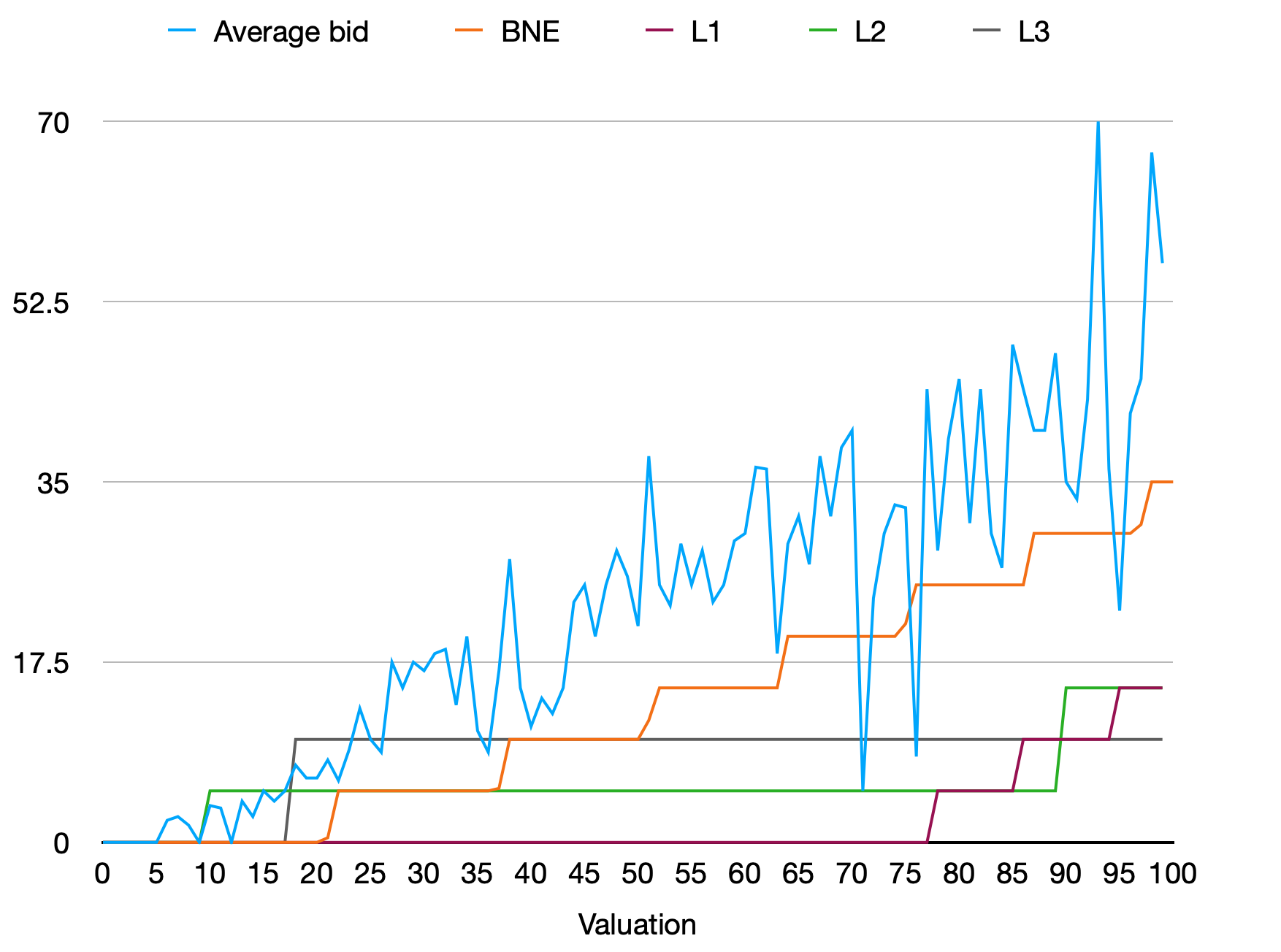}
    \label{risk_fives}
\end{figure}

\end{appendices}

\end{document}